\documentclass{article}

\usepackage{hyperref}
\usepackage{amsmath}
\usepackage{amsthm}
\usepackage{amssymb}
\usepackage{xcolor}
\usepackage{booktabs}

\usepackage[utf8]{inputenc}

\usepackage{modello} 
\usepackage{tikz}
\usetikzlibrary{calc,
                positioning}

\begin{document}
\title{Constructions of Linear Codes from Vectorial Plateaued Functions  and Their Subfield Codes with Applications to Quantum CSS Codes}

\author{Virginio Fratianni and Sihem Mesnager}
\date{Université Paris 8, Laboratoire de Géométrie, Analyse et Applications, LAGA, Université Sorbonne Paris Nord, CNRS, UMR 7539, France.
\\ \emph{Email addresses}: virginio.fratianni@etud.univ-paris8.fr, smesnager@univ-paris8.fr}

\maketitle            



\section*{Abstract}
Linear codes over finite fields parameterized by functions have proven to be a powerful tool in coding theory, yielding optimal and few-weight codes with significant applications in secret sharing, authentication codes, and association schemes. In 2023, Xu et al. introduced a construction framework for 3-dimensional linear codes parameterized by two functions, which has demonstrated considerable success in generating infinite families of optimal linear codes. Motivated by this approach, we propose a construction that extends the framework to three functions, thereby enhancing the flexibility of the parameters. Additionally, we introduce a vectorial setting by allowing vector-valued functions, expanding the construction space and the set of achievable structural properties.

We analyze both scalar and vectorial frameworks, employing Bent and $s$-Plateaued functions, including Almost Bent, to define the code generators. By exploiting the properties of the Walsh transform, we determine the explicit parameters and weight distributions of these codes and their punctured versions. A key result of this study is that the constructed codes have few weights, and their duals are distance and dimensionally optimal with respect to both the Sphere Packing and Griesmer bounds.

Furthermore, we establish a theoretical connection between our vectorial approach and the classical first generic construction of linear codes, providing sufficient conditions for the resulting codes to be minimal and self-orthogonal. Finally, we investigate applications to quantum coding theory within the Calderbank–Shor–Steane framework.
\\

\textbf{Keywords.} 
Vectorial functions, plateaued functions, (almost) bent functions, subfield codes, secret sharing scheme.\\

\textbf{Mathematics Subject Classification.} 11T71, 14G50, 94A60 

\section{Introduction}
Coding theory examines the properties of codes and their suitability for specific applications, including data compression for cloud storage, wireless data transmission, and cryptography. It involves various scientific disciplines, such as mathematics, information theory, computer science, artificial intelligence, and linguistics; however, it is grounded in mathematical methods, particularly algebraic methods in the linear code case. For instance, we have many connections with number theory in the context of Gauss sums, cyclotomic field theory, Galois fields and tools from the theory of Fourier transform and the related exponential sums.
\\\\
Today, coding theory is of high importance because many infrastructures, such as wireless communications and cloud storage, rely on codes. To a great extent, the research on coding theory is based on mathematics, since nearly all the theoretical bases, constructions, and analyses of codes are based on mathematical methods.
\\\\
Subfield codes were initially introduced in \cite{ref6, ref7} as a method to derive optimal linear codes over 
\( \mathbb{F}_q \) from linear codes defined over \( \mathbb{F}_{q^m} \). Given any linear code \( C \) over 
\( \mathbb{F}_{q^m} \), its corresponding subfield code \( C^{(q)} \) can exhibit either favorable or unfavorable parameters. 
To ensure good subfield codes, researchers have carefully selected suitable codes over 
\( \mathbb{F}_{q^m} \) and developed well-structured codes over \( \mathbb{F}_q \).

Ding and Heng focused on elliptic quadric codes and Tits ovoid codes, analyzing the parameters and weight distributions 
of their subfield codes \cite{ref14}. Later, their research expanded to include the 
subfield codes of hyperoval codes, conic codes, maximal arcs, and certain MDS codes with dimension 2 
\cite{ref17, ref18, ref19, ref20}. These subfield codes often exhibit optimal or nearly optimal characteristics and are 
notable for their limited number of distinct weights.

Further investigations into subfield codes with advantageous parameters and few weights have been conducted in 
\cite{DingZhu, Liu, Qiao, ref31, ref33, ref34, ref35, Xie} and the references therein, contributing to a broader understanding of their properties 
and applications.
\\\\
In 2023, Xu et al. \cite{xu} presented a construction of $3$-dimensional linear codes involving two functions $f$ and $g$ and determined the parameters of the code, the punctured ones and the subfield ones. They also analyzed this family of codes for specific functions, such as the trace, the norm or bent ones.\\
Motivated by their work, we propose a construction method of $4$-dimensional linear codes involving $3$ functions.\\
Let $f$, $g$ and $h$ be three functions from $\mathbb{F}_{q^m}$ to $\mathbb{F}_q$, where $q$ is a power of a prime $p$. We define the subset $\mathcal{D}$ of $\mathbb{F}_{q^m}^3$ as
\begin{equation}\label{defining set}
    \mathcal{D}=\left\{(x, y,z) \in \mathbb{F}_{q^m}^3: f(x)+g(y)+h(z)=0\right\}
\end{equation}
and a $4 \times(\# \mathcal{D})$ matrix $G_{f, g,h}^*$ over $\mathbb{F}_{q^m}$ as 
\begin{equation}\label{punctured matrix}
    G_{f, g,h}^*=\left(\begin{array}{l} 1 \\ x \\ y \\ z \end{array}\right)_{(x, y,z) \in \mathcal{D}}.
\end{equation}
Let $\mathcal{C}_{f,g,h}$ be the $[\# \mathcal{D}+1,4]$ linear code over $\mathbb{F}_{q^m}$ generated by the following matrix:
\begin{equation}\label{matrix}
    G_{f, g,h}=\left(\begin{array}{cc} 0 & \\ 1 & G_{f, g,h}^* \\ 0 \\ 0& \end{array}\right).
\end{equation}
Let $\mathcal{C}_{f,g,h}^*$ be the punctured code obtained from $\mathcal{C}_{f,g,h}$ by puncturing on the first coordinate, hence $\mathcal{C}_{f,g,h}^*$ is the $[\# \mathcal{D},4]$ linear code over $\mathbb{F}_{q^m}$ with generator matrix $G_{f, g,h}^*$.\\\\
Our article focuses on deepening the relationship between coding theory and symmetric cryptography, in the context of codes built using cryptographic functions. This framework constitutes a very active research area nowadays, since the resulting codes often have good parameters and few weights. We refer the reader, for instance, to \cite{Fratianni, ChapterMesnager, MS2, MS1, Sinak} and the literature cited therein.
\\\\
The main contributions of this work can be summarized as follows:

\begin{itemize}
    \item \textbf{Code construction from $p$-ary $s$-plateaued functions:} 
    We analyze a specific instance of the proposed code construction involving the relative trace, the trace of the square, and a generic function $h$. By specializing $h$ to be a $s$-plateaued $p$-ary function, we explicitly determine the length, dimension, and weight distribution of the resulting subfield codes and their punctured versions. We prove that these are four-weight and five-weight linear codes.

    \item \textbf{Vectorial framework and $s$-plateaued functions:} 
    We extend the construction to the vectorial setting, considering codes generated by matrices defined via two functions $f$ and $g$. We derive a general formula for the weight distribution linking the code parameters to the Walsh transform of the composition $h = f^{-1} \circ g$. We specifically analyze the case where $h$ is an $s$-plateaued function (including the important class of Almost Bent functions where $s=1$), obtaining a family of few-weight codes.

    \item \textbf{Optimality of codes:} 
    A significant result of this study is that the constructed codes have few weights. In addition, the duals of the codes constructed from $s$-plateaued $p$-ary functions, with odd $s$, are distance and dimensionally optimal with respect to the Sphere Packing bound. In the vectorial framework, Almost Bent functions induce codes that are almost dimensionally optimal with respect to the Sphere Packing bound. In the remaining cases, the dual codes are almost dimensionally or distance optimal with respect to both the Sphere Packing and Griesmer bounds.
    
    \item \textbf{Link to generic constructions and minimality:} 
    We establish a theoretical connection between our vectorial construction and the classical first generic construction of linear codes. We show that our codes can be viewed as augmented versions of codes derived from the inverse of plateaued functions. Furthermore, we provide sufficient conditions for these codes to be minimal (satisfying the Ashikhmin-Barg bound) and self-orthogonal.

    \item \textbf{Applications to quantum codes:}
    We apply the proposed linear codes to fault-tolerant quantum computation, showing that codes derived from vectorial plateaued functions naturally fit into the Calderbank–Shor–Steane framework and support transversal phase gates. In the binary case, appropriate $s$-plateaued functions produce doubly-even codes, recovering CSS$_T$ codes with transversal $T$ gates.
\end{itemize}

The rest of the paper is organized as follows.
In Section~2, we recall preliminaries on linear codes, subfield codes, functions over finite fields, and character theory.
Section~3 is devoted to the determination of the parameters of the codes $\mathcal{C}_{f,g,h}$,
their punctured versions, and their subfield codes for generic functions
$f,g,h$ from $\mathbb{F}_{q^m}$ to $\mathbb{F}_q$.
In Section~4, we apply these results to specific families of functions, including bent ones.
Section~5 extends the construction to the vectorial setting, considering functions from
$\mathbb{F}_{q^m}$ to itself, with particular emphasis on the binary case. In Section~6, we investigate applications to quantum coding theory.
We show that linear codes arising from vectorial plateaued functions fit naturally into the
Calderbank--Shor--Steane framework and yield quantum codes admitting transversal phase gates.
In particular, in the binary case, suitably chosen $s$-plateaued functions give rise to
CSS$_T$ codes with transversal $T$-gates, recovering known constructions from a unified
Walsh spectral perspective.
Beyond characteristic two, and in particular in the ternary setting, the same approach
leads to CSS phase codes admitting transversal quadratic phase gates, producing new families
of quantum codes not derived from classical Reed--Muller constructions. Finally, Section~7 concludes the paper and outlines perspectives for future research.

\section{Preliminaries}

Throughout this paper, unless otherwise stated, codes constructed from functions have length a multiple of $N=q^m$, where $q$ is a power of a prime $p$, corresponding to evaluations
over the entire field $\mathbb{F}_{q^m}$.

In this section, we recall some general results on linear codes, subfield codes, character theory, and Boolean and vectorial functions that will be needed throughout the paper.

\subsection{Coding theory}
An $[n, k, d]$ linear code $\mathcal{C}$ over $\mathbb{F}_{q}$ is a linear subspace of $\mathbb{F}_{q}^n$ with dimension $k$ and minimum Hamming distance $d$.

The Hamming distance between two vectors $x = (x_0, \dots, x_{n-1})$ and $y = (y_0, \dots, y_{n-1})$ in $\mathbb{F}_{q}^n$, denoted by $d_H(x, y)$, is defined as the number of coordinate positions in which they differ. Formally:
$$
d_H(x, y) = \left| \{ i \in \{0, \dots, n-1\} : x_i \neq y_i \} \right|
$$
Consequently, the minimum distance $d$ of the linear code $\mathcal{C}$ is defined as:
$$
d = \min \{ d_H(c, c') : c, c' \in \mathcal{C}, c \neq c' \}
$$
Since $\mathcal{C}$ is linear, the minimum distance is equal to the minimum Hamming weight of the non-zero codewords:
$$
d = \min \{ w_H(c) : c \in \mathcal{C}, c \neq 0 \}
$$
where $w_H(c)$ is the number of non-zero entries of $c$.

The value $n-k$ is called the codimension of $\mathcal{C}$. Given a linear code $\mathcal{C}$ of length $n$ over $\mathbb{F}_{q}$, its Euclidean dual code is denoted by $\mathcal{C}^{\perp}$. It is defined by
$$
\mathcal{C}^{\perp}=\left\{\left(b_{0}, b_{1}, \ldots, b_{n-1}\right) \in \mathbb{F}_{q}^{n}\,:\, \sum_{i=0}^{n-1} b_{i} c_{i}=0 \,\,\text{for every}\,\left(c_{0}, c_{1}, \ldots, c_{n-1}\right) \in \mathcal{C}\right\}.
$$

The minimum distance of an $[n, k, d]$ linear code is limited by the Singleton bound
$$
d \leq n-k+1 .
$$
A code meeting the above bound is called Maximum Distance Separable (MDS). We will also need the following two classical bounds from coding theory \cite{Huffman}, with respect to which we will study the optimality of our codes.
\begin{lemma}[Sphere Packing / Hamming Bound]
Let $\mathcal{C}$ be an $[n, k, d]$ code over $\mathbb{F}_q$. Then
$$
q^k \sum_{i=0}^{\lfloor \frac{d-1}{2} \rfloor} \binom{n}{i} (q-1)^i \leq q^n.
$$
\end{lemma}

\begin{lemma}[Griesmer Bound]
Let $\mathcal{C}$ be an $[n, k, d]$ linear code over $\mathbb{F}_q$ with $k\ge 1$. Then
$$
n \geq \sum_{i=0}^{k-1} \left\lceil \frac{d}{q^i} \right\rceil.
$$
\end{lemma}

In our computations, we will also need the \emph{Pless Power Moments} (see \cite{Huffman}), which provide a set of identities relating the moments of the weight distribution of a linear code to those of its dual. Let $(A_0, A_1, \dots, A_n)$ be the weight distribution of the code $\mathcal{C}$, where $A_i$ is the number of codewords of weight $i$. Similarly, let $(B_0, B_1, \dots, B_n)$ be the weight distribution of the dual code $\mathcal{C}^{\perp}$. The first four identities are explicitly given by:

$$ \sum_{i=0}^{n} A_i = q^k, $$

$$ \sum_{i=0}^{n} i A_i = q^{k-1} [ n(q-1) - B_1 ], $$

$$ \sum_{i=0}^{n} i^2 A_i = q^{k-2} \left[ n(q-1)(n(q-1)+1) - \big(2(n-1)(q-1) + q\big) B_1 + 2 B_2 \right], $$

\begin{align*}
\sum_{i=0}^{n} i^3 A_i = q^{k-3} \Big[ & n(q-1) \big( n^2(q-1)^2 + 3n(q-1) + 1 \big) \\
&- \Big( 3(n-1)(q-1)\big(n(q-1)+1\big) + q(3n(q-1)+1) \Big) B_1 \\
&+ \big( 6(n-2)(q-1) + 6q \big) B_2 - 6 B_3 \Big].
\end{align*}
Given a linear code \( \mathcal{C} \) of length \( n \) and dimension \( k \) over \( \mathbb{F}_{q^m} \) with a generator matrix \( G \), we can construct the subfield code \( \mathcal{C}^{(q)} \) of length \( n \) and dimension \( k' \) over \( \mathbb{F}_q \).

To do this, first take a basis for \( \mathbb{F}_{q^m} \) over \( \mathbb{F}_q \). Then, replace each entry of the generator matrix \( G \) with the corresponding \( m \times 1 \) column vector from \( \mathbb{F}_{q^m} \) relative to this basis. As a result, \( G \) transforms into a \( km \times n \) matrix over \( \mathbb{F}_q \), which generates the new subfield code \( \mathcal{C}^{(q)} \) over \( \mathbb{F}_q \).

It is evident that the dimension \( k' \) of \( \mathcal{C}^{(q)} \) satisfies \( k' \leq mk \). Ding et al. \cite{ref14} demonstrated that this subfield code is independent of both the choice of \( G \) and the basis for \( \mathbb{F}_{q^m} \) over \( \mathbb{F}_q \). Furthermore, they showed that the subfield code is equivalent to the trace code, leading to the conclusion that \( \mathcal{C}^{(q)} = \operatorname{Tr}_{q^m/q}(\mathcal{C}) \).

\begin{lemma} [\cite{ref14}, Theorem 2.5]
Let \( \mathcal{C} \) be a linear code of length \( n \) over \( \mathbb{F}_{q^m} \) with a generator matrix \( G = \left[g_{ij}\right]_{1 \leq i \leq k, 1 \leq j \leq n} \). The trace representation of the subfield code \( \mathcal{C}^{(q)} \) is given by 
\begin{equation*}
\mathcal{C}^{(q)} = \bigg\{ \Big(
 \operatorname{Tr}_{q^m / q}\Big(\sum_{i=1}^k a_i g_{i1}\Big), \ldots, \operatorname{Tr}_{q^m / q}\Big(\sum_{i=1}^k a_i g_{in}\Big) \Big) : 
 a_1, \ldots, a_k \in \mathbb{F}_{q^m} 
\bigg\},
\end{equation*}
where \( \operatorname{Tr}_{q^m / q}(x) = \sum_{i=0}^{m-1} x^{q^i} \) denotes the trace function from \( \mathbb{F}_{q^m} \) to \( \mathbb{F}_q \).
\end{lemma}
\begin{lemma}[\cite{ref14}]
    The minimal distance $d^{\perp}$ of $\mathcal{C}^{\perp}$ and the minimal distance $d^{(q)\perp}$ of $\mathcal{C}^{(q)\perp}$ satisfy
$$
d^{(q)\perp} \geq d^{\perp}.
$$
\end{lemma}

\subsection{Vectorial and scalar functions over finite fields}

     Throughout this paper, we regard the finite field $\mathbb{F}_{p^n}$ as an $n$-dimensional vector space over its prime subfield $\mathbb{F}_p$. We shall use the notation $\mathbb{F}_{p}^n$ to emphasize the vector space structure, reserving $\mathbb{F}_{p^n}$ for contexts where the multiplicative field properties are essential.

Accordingly, the inner product is defined depending on the representation used:
\begin{itemize}
    \item In the vector space $\mathbb{F}_p^n$, we denote by ``$\cdot$'' the standard Euclidean scalar product. For any two vectors $x = (x_1, \dots, x_n)$ and $y = (y_1, \dots, y_n)$ in $\mathbb{F}_p^n$, this is given by:
    $$
    x \cdot y = \sum_{i=1}^n x_i y_i.
    $$
    \item Under the field identification $\mathbb{F}_{p^n} \cong \mathbb{F}_p^n$, the role of the inner product is assumed by the trace form. For any pair of elements $\alpha, \beta \in \mathbb{F}_{p^n}$, the duality is induced by the absolute trace function $\text{Tr}_{p^n/p}: \mathbb{F}_{p^n} \to \mathbb{F}_p$, defined as:
    $$
    \langle \alpha, \beta \rangle = \text{Tr}_{p^n/p}(\alpha \beta).
    $$
\end{itemize}
Unless otherwise stated, we implicitly identify the standard dot product with the trace form when moving between these representations.

\begin{defn}
Let $q=p^r$ where $p$ is a prime. A vectorial function
 $\mathbb{F}_q^n\rightarrow \mathbb{F}_q^m$ (or $\mathbb{F}_{q^n}\rightarrow \mathbb{F}_{q^m}$) is called an \emph{$(n,m)$-q-ary function}. When $q=2$, an $(n,m)$-2-ary function will be simply denoted an \emph{$(n,m)$-function}. A \emph{Boolean function} is an $(n,1)$-function, that is, a function $\mathbb{F}_2^n\rightarrow \mathbb{F}_2$ (or $\mathbb{F}_{2^n}\rightarrow \mathbb{F}_2$).
\end{defn}

Consider a scalar function $f: \mathbb{F}_{q^n} \to \mathbb{F}_q$. The \emph{Walsh transform} of $f$ at a point $\lambda \in \mathbb{F}_{q^n}$ is defined as:
$$
\mathcal{W}_f(\lambda) = \sum_{x \in \mathbb{F}_{q^n}} \zeta_p^{\text{Tr}_{q/p}(f(x)) - \text{Tr}_{q^n/p}(\lambda x)},
$$
where $\zeta_p = e^{2\pi i / p}$ is a primitive $p$-th root of unity.

The \emph{Walsh spectrum} is the set of values $\{ \mathcal{W}_f(\lambda) : \lambda \in \mathbb{F}_{q^n} \}$. Based on this, $f$ is classified as follows:

\begin{defn}[Bent, Plateaued and Almost Bent Scalar Functions]
Let $f: \mathbb{F}_{q^n} \to \mathbb{F}_q$.
\begin{enumerate}
    \item \textbf{Bent:} The function $f$ is called \textit{bent} if $|\mathcal{W}_f(\lambda)|^2 = q^n$ for all $\lambda \in \mathbb{F}_{q^n}$. This is possible only if $n$ is even (or $n$ odd and $q$ a square).
    
    \item \textbf{$s$-Plateaued:} The function $f$ is called \textit{$s$-plateaued} (for an integer $0 \leq s \leq n$) if the squared magnitude of its Walsh transform takes only two values:
    $$
    |\mathcal{W}_f(\lambda)|^2 \in \{0, q^{n+s}\} \quad \text{for all } \lambda \in \mathbb{F}_{q^n}.
    $$
    
    \item \textbf{Almost Bent:} While the term is standard for vectorial functions, in the scalar case, a function is often informally referred to as \textit{almost bent} (or semi-bent/near-bent) if it is $s$-plateaued with the minimal possible $s > 0$. For $n$ odd, this corresponds to $s=1$.
\end{enumerate}
\end{defn}

Let $F: \mathbb{F}_{q^n} \to \mathbb{F}_{q^m}$ be a vectorial function. Its cryptographic properties are analyzed through its \emph{component functions}. For any non-zero element $v \in \mathbb{F}_{q^m}^*$, the component function $f_v: \mathbb{F}_{q^n} \to \mathbb{F}_q$ is defined as:
$$
f_v(x) = \text{Tr}_{q^m/q}(v F(x)).
$$
The \emph{Walsh transform} of $F$ at $(v, \lambda) \in \mathbb{F}_{q^m}^* \times \mathbb{F}_{q^n}$ is defined as the Walsh transform of the scalar function $f_v$ evaluated in $\lambda$:
$$
\mathcal{W}_F(v,\lambda) = \sum_{x \in \mathbb{F}_{q^n}} \zeta_p^{\text{Tr}_{q^m/p}(vF(x)) - \text{Tr}_{q^n/p}(\lambda x)}.
$$

\begin{defn}[Vectorial Plateaued and Almost Bent Functions]
Let $F: \mathbb{F}_{q^n} \to \mathbb{F}_{q^m}$.
\begin{enumerate}
    \item \textbf{$s$-Plateaued:} The function $F$ is called \textit{vectorial $s$-plateaued} if for every non-zero $v \in \mathbb{F}_{q^m}^*$, the component function $x \mapsto \text{Tr}_{q^m/q}(vF(x))$ is an $s$-plateaued scalar function. Explicitly:
    $$
    |\mathcal{W}_F(\lambda, v)|^2 \in \{0, q^{n+s}\} \quad \text{for all}\,\, \lambda \in \mathbb{F}_{q^n}, v \in \mathbb{F}_{q^m}^*.
    $$
    
    \item \textbf{Almost Bent:} This definition applies specifically when $n=m$ and $n$ is odd. The function $F$ is called \textit{Almost Bent} if it achieves the minimum possible non-linearity for a non-bent function. This occurs when the Walsh spectrum is as flat as possible, specifically:
    $$
    \mathcal{W}_F(\lambda, v) \in \{0, \pm q^{\frac{n+1}{2}}\} \quad \text{for all}\,\, \lambda \in \mathbb{F}_{q^n}, v \in \mathbb{F}_{q^n}^*.
    $$
\end{enumerate}
\end{defn}
 \begin{remark*}
        In terms of the plateaued classification, an Almost Bent function is exactly a vectorial \emph{$1$-plateaued} function ($s=1$).
    \end{remark*}
Scalar and vectorial functions over finite fields are fundamental mathematical objects in modern cryptography, coding theory, and combinatorics. For a general treatment of the subject, we refer the reader to \cite{CarletPlateaued, Carlet-book, Mesnager-book}.

\subsection{Character theory over finite fields}

Let $\mathbb{F}_{q}$ be a finite field with $q = p^r$ elements, where $p$ is a prime number and $r$ is a positive integer. We denote by $\mathbb{F}_{q}^* = \mathbb{F}_{q} \setminus \{0\}$ the multiplicative group of the field.

Let $\text{Tr}: \mathbb{F}_{q} \to \mathbb{F}_{p}$ denote the absolute trace function defined by $\text{Tr}(x) = x + x^p + \dots + x^{p^{r-1}}$.
An \textit{additive character} of $\mathbb{F}_{q}$ is a homomorphism $\chi$ from the additive group $(\mathbb{F}_{q}, +)$ to the multiplicative group of complex numbers of modulus 1. For any $a \in \mathbb{F}_{q}$, the function $\chi_a: \mathbb{F}_{q} \to \mathbb{C}^*$ defined by
$$
\chi_a(x) = \zeta_p^{\text{Tr}(ax)},
$$
is an additive character, where $\zeta_p = e^{2\pi i / p}$ is a primitive $p$-th root of unity.
The character $\chi_0$ (where $a=0$) is called the \textit{trivial additive character}, satisfying $\chi_0(x) = 1$ for all $x \in \mathbb{F}_{q}$. The character $\chi_1$ (where $a=1$) is referred to as the \textit{canonical additive character}.
The orthogonality relation for additive characters is given by:
$$
\sum_{x \in \mathbb{F}_{q}} \chi_a(x) = \begin{cases} q & \text{if } a = 0, \\ 0 & \text{if } a \neq 0. \end{cases}
$$

A \textit{multiplicative character} of $\mathbb{F}_{q}$ is a homomorphism $\psi$ from the multiplicative group $\mathbb{F}_{q}^*$ to the complex unit group. Let $g$ be a fixed primitive element of $\mathbb{F}_{q}$. Any multiplicative character $\psi_j$ (for $0 \leq j \leq q-2$) can be defined by
$$
\psi_j(g^k) = \zeta_{q-1}^{jk} \quad \text{for } k = 0, 1, \dots, q-2,
$$
where $\zeta_{q-1} = e^{2\pi i / (q-1)}$. We extend the domain of $\psi$ to all of $\mathbb{F}_{q}$ by setting $\psi(0) = 0$ (except for the trivial character $\psi_0$, where $\psi_0(0)=1$ is sometimes defined, though usually $\psi_0(x)=1$ for $x \neq 0$).
For odd $q$, the unique multiplicative character of order 2 is denoted by $\eta$ and is called the \textit{quadratic character} of $\mathbb{F}_{q}$. It is defined as $\eta(x) = 1$ if $x$ is a square in $\mathbb{F}_{q}^*$, and $\eta(x) = -1$ if $x$ is a non-square.

The orthogonality relation for multiplicative characters is given by:
$$
\sum_{x \in \mathbb{F}_{q}^*} \psi_j(x) = \begin{cases} q-1 & \text{if } j = 0, \\ 0 & \text{if } j \neq 0. \end{cases}
$$

Let $\psi$ be a multiplicative character and $\chi$ be an additive character of $\mathbb{F}_{q}$. The \textit{Gauss sum} $G(\psi, \chi)$ is defined by
$$
G(\psi, \chi) = \sum_{c \in \mathbb{F}_{q}^*} \psi(c) \chi(c).
$$
For the specific case of the quadratic character $\eta$ and the canonical additive character $\chi_1$, the values of the quadratic Gauss sums are well-known.

\begin{lemma}[\cite{Lidl}]
Let $q = p^r$ where $p$ is an odd prime. The quadratic Gauss sum $G(\eta, \chi_1)$ is explicitly given by:
$$
G(\eta, \chi_1) = \begin{cases}
(-1)^{r-1} \sqrt{q} & \text{if } p \equiv 1 \pmod 4, \\
(-1)^{r-1} i^r \sqrt{q} & \text{if } p \equiv 3 \pmod 4.
\end{cases}
$$
\end{lemma}

Character sums involving polynomials are often referred to as \emph{Weil sums}. The evaluation of such sums for quadratic polynomials is essential for our results.

\begin{lemma}[\cite{Lidl}]\label{Character and quadratic polynomial}
Let $\chi_b$ be a nontrivial additive character of $\mathbb{F}_{q}$ (with $b \in \mathbb{F}_{q}^*$) and let $f(x) = a_2 x^2 + a_1 x + a_0 \in \mathbb{F}_{q}[x]$ with $a_2 \neq 0$.
The Weil sum $\sum_{c \in \mathbb{F}_{q}} \chi_b(f(c))$ satisfies the following:
\begin{enumerate}
    \item If $q$ is odd, then:
    $$
    \sum_{c \in \mathbb{F}_{q}} \chi_b(f(c)) = \chi_b\left(a_0 - a_1^2 (4a_2)^{-1}\right) \eta(a_2) G(\eta, \chi_b).
    $$
    \item If $q$ is even, then:
    $$
    \sum_{c \in \mathbb{F}_{q}} \chi_b(f(c)) = \begin{cases}
    \chi_b(a_0) q & \text{if } a_2 = b a_1^2, \\
    0 & \text{otherwise}.
    \end{cases}
    $$
\end{enumerate}
\end{lemma}

\section{Subfield codes from generic functions}
In this section, we present the parameters of the codes $\mathcal{C}_{f,g,h}$ in construction \eqref{matrix} for generic functions $f,g,h$ from $\mathbb{F}_{q^m}$ to $\mathbb{F}_q$, their punctured codes $\mathcal{C}_{f,g,h}^*$ as in \eqref{punctured matrix} and their subfield codes $\mathcal{C}_{f,g,h}^{(q)}$.\\\\
In this context, the trace representation of the subfield code $\mathcal{C}_{f, g,h}^{(q)}$ is given by $$ \mathcal{C}_{f, g,h}^{(q)}=\left\{\mathbf{c}_{a, b, c,d}: a \in \mathbb{F}_q, b, c, d \in \mathbb{F}_{q^m}\right\}, $$ where $$ \mathbf{c}_{a, b, c,d}=\left(\operatorname{Tr}_{q^m / q}(b),\left(a+\operatorname{Tr}_{q^m / q}(b x+c y+dz)\right)_{(x, y,z) \in \mathcal{D}}\right). $$
As in standard notation, let $\eta$ and $\eta '$ be the quadratic
characters of $\mathbb{F}_q$ and $\mathbb{F}_{q^m}$, respectively. Also, let $\chi$ and $\chi '$ be the canonical additive characters of $\mathbb{F}_q$ and $\mathbb{F}_{q^m}$, respectively. Hence, $\chi'=\chi\circ\text{Tr}_{q^m/q}$. Consider three functions $f,g$ and $h$ from $\mathbb{F}_{q^m}$ to $\mathbb{F}_q$.\\
We will need the following two terms, depending on the functions.
\begin{equation}\label{Phi}
\Phi_{f,g,h}=\sum_{s\in\mathbb{F}_q^*}\sum_{x\in\mathbb{F}_{q^m}}\chi(sf(x))\sum_{y\in\mathbb{F}_{q^m}}\chi(sg(y))\sum_{z\in\mathbb{F}_{q^m}}\chi(sh(z)).
\end{equation}
Then, for $(a,b,c,d)\in\mathbb{F}_q\times\mathbb{F}_{q^m}^3$:
\begin{equation}\label{lambda}
    \Lambda_{a,b,c,d}^{f,g,h}=\sum_{s\in\mathbb{F}_q^*}\chi(sa)\sum_{\omega\in\mathbb{F}_q^*}\sum_{(x,y,z)\in\mathbb{F}_{q^m}^3}\chi(\omega f(x)+\omega g(y)+\omega h(z))\chi'(sbx+scy+sdz).
\end{equation}

In the following lemmas, we will compute the length and weights of the codes with respect to the latter terms.

\begin{lemma}\label{lemma: length}
    The length of the codes $\mathcal{C}_{f,g,h}$ and $\mathcal{C}_{f,g,h}^{(q)}$ is $n=1+q^{3m-1}+\frac{1}{q}\Phi_{f,g,h}$.
\end{lemma}
\begin{proof}
    Let $n$ be the length of the codes. From construction \eqref{matrix} it is clear that $n=\# \mathcal{D}+1$, due to the first column of the generator matrix. Hence, we compute the cardinality of the defining set $\mathcal{D}$ using the well-known results of orthogonality in character theory.
    \begin{align}
        \# \mathcal{D}&=\frac{1}{q}\sum_{(x,y,z)\in\mathbb{F}_{q^m}^3}\sum_{s\in\mathbb{F}_q}\chi(s(f(x)+g(y)+h(z)))\notag\\
        &=q^{3m-1}+\frac{1}{q}\sum_{s\in\mathbb{F}_q^*}\sum_{(x,y,z)\in\mathbb{F}_{q^m}^3}\chi(s(f(x)+g(y)+h(z)))\notag\\
        &=q^{3m-1}+\frac{1}{q}\sum_{s\in\mathbb{F}_q^*}\sum_{x\in\mathbb{F}_{q^m}}\chi(sf(x))\sum_{y\in\mathbb{F}_{q^m}}\chi(sg(y))\sum_{z\in\mathbb{F}_{q^m}}\chi(sh(z))\notag\\
        &=q^{3m-1}+\frac{1}{q}\Phi_{f,g,h}.\notag
    \end{align}
\end{proof}
Now we consider the weights $\text{wt}(\textbf{c}_{a,b,c,d})$, that is, the weights of the codewords in the subfield code $\mathcal{C}_{f,g,h}^{(q)}$. Let $\delta:\mathbb{F}_{q^m}\rightarrow \{0,1\}$ be a functions such that 
\begin{equation}\label{delta}
    \delta(x)=\begin{cases}
        0 \,\,\,\,\,\text{if} \,\,\text{Tr}_{q^m/q}(x)=0\\
        1 \,\,\,\,\,\text{otherwise}.
    \end{cases}\notag
\end{equation}

\begin{lemma}\label{lemma: weights}
    For any $(a,b,c,d)\in\mathbb{F}_q\times\mathbb{F}_{q^m}^3$, the weight of the codeword $\textbf{c}_{a,b,c,d}$ in $\mathcal{C}_{f,g,h}^{(q)}$ is given by:
    $$\text{wt}(\textbf{c}_{a,b,c,d})=\begin{cases}
        0\,\,\,\,\,\,\,\,\,\,\,\,\,\,\,\,\,\,\,\,\,\,\,\,\,\,\,\,\,\,\,\,\,\,\,\,\,\,\,\,\,\,\,\,\,\,\,\,\,\,\,\,\,\,\,\,\,\,\,\,\,\,\,\,\,\,\,\,\,\,\,\,\,\,\,\,\,\,\,\,\,\,\,\,\,\,\,\,\,\,\,\,\,\,\,\,\,\,\,\,\,\,\,\,\,\,\,\,\,\,\,\,\,\,\,\,\,\,\,\,\,\,\,\,\,\,\,\,\,\,\,\text{if} \,\,a=b=c=d=0,\\
        q^{3m-1}+\frac{1}{q}\Phi_{f,g,h} \,\,\,\,\,\,\,\,\,\,\,\,\,\,\,\,\,\,\,\,\,\,\,\,\,\,\,\,\,\,\,\,\,\,\,\,\,\,\,\,\,\,\,\,\,\,\,\,\,\,\,\,\,\,\,\,\,\,\,\,\,\,\,\,\,\,\,\,\,\,\,\,\,\,\,\,\,\,\,\,\,\,\,\,\,\,\,\,\,\,\,\, \text{if}\,\, a\not =0\,\,\text{and}\,\,b=c=d=0\\
        \delta(b)+q^{3m-2}(q-1)+\frac{1}{q^2}\big[ (q-1)\Phi_{f,g,h}-\Lambda_{a,b,c,d}^{f,g,h} \big]\,\,\,\,\,\,\text{otherwise}.  
    \end{cases}$$
\end{lemma}
\begin{proof}
    For $a\in\mathbb{F}_q^*$ and $b,c,d\in\mathbb{F}_{q^m}$ we define the term $N_{a,b,c,d}$ as the cardinality of the following set.
    $$N_{a,b,c,d}=\#\{(x,y,z)\in\mathcal{D}\,:\,a+\text{Tr}_{q^m/q}(bx+cy+dz)=0\}.$$
    Clearly $N_{0,0,0,0}=\#\mathcal{D}$ and $N_{a,0,0,0}=0$ for $a\in\mathbb{F}_q^*$. For $b,c,d$ not all zero we have:
    \begin{align}
        qN_{a,b,c,d}&=\sum_{(x,y,z)\in\mathcal{D}}\sum_{s\in\mathbb{F}_q}\chi(s(a+\text{Tr}_{q^m/q}(bx+cy+dz)))\notag\\
    &=\#\mathcal{D}+\sum_{(x,y,z)\in\mathcal{D}}\sum_{s\in\mathbb{F}_q^*}\chi(s(a+\text{Tr}_{q^m/q}(bx+cy+dz)))\notag\\
&=\#\mathcal{D}+\sum_{(x,y,z)\in\mathcal{D}}\sum_{s\in\mathbb{F}_q^*}\chi(sa)\chi'(sbx+scy+sdz)\notag\\
    &=\#\mathcal{D}+\sum_{s\in\mathbb{F}_q^*}\chi(sa)\sum_{(x,y,z)\in\mathbb{F}_{q^m}^3}\big(\frac{1}{q}\sum_{\omega\in\mathbb{F}_q}\chi(\omega(f(x)+g(y)+h(z))\big )\chi'(sbx+scy+sdz)\notag\\
&=\#\mathcal{D}+\frac{1}{q}\sum_{s\in\mathbb{F}_q^*}\chi(sa)\sum_{(x,y,z)\in\mathbb{F}_{q^m}^3}\chi'(sbx+scy+sdz)+\frac{1}{q}\Lambda_{a,b,c,d}^{f,g,h}\notag\\
&=\#\mathcal{D}+\frac{1}{q}\Lambda_{a,b,c,d}^{f,g,h}.\notag
    \end{align}
    The last equality is due to the fact that, for $s\in\mathbb{F}_q^*$ and $b,c,d$ not all zero:
    $$\sum_{(x,y,z)\in\mathbb{F}_{q^m}^3}\chi'(sbx+scy+sdz)=$$
    $$=\big(\sum_{x\in\mathbb{F}_{q^m}}\chi'(sbx)\big)\big(\sum_{y\in\mathbb{F}_{q^m}}\chi'(scy)\big)\big(\sum_{z\in\mathbb{F}_{q^m}}\chi'(sdz)\big)=0.$$
    Finally, from the equality 
    \begin{equation}\label{weight formula}
        \text{wt}(\textbf{c}_{a,b,c,d})=\delta(b)+\#\mathcal{D}-N_{a,b,c,d}
    \end{equation}
    we derive our statement.
\end{proof}

Since $\mathcal{C}^*_{f,g,h}$ is the punctured code of $\mathcal{C}_{f,g,h}$ at the first coordinate, from the two last lemmas we can deduce similar results for the punctured codes.

\begin{corollary}\label{corollary: parameters of the punctured codes}
    The length of the codes $\mathcal{C}_{f,g,h}^*$ and $\mathcal{C}_{f,g,h}^{*(q)}$ is $n=q^{3m-1}+\frac{1}{q}\Phi_{f,g,h}$. In addition,
    $$\text{wt}(\textbf{c}^*_{a,b,c,d})=\text{wt}(\textbf{c}_{a,b,c,d})-\delta(b),$$
    where the codeword $\textbf{c}^*_{a,b,c,d}$ is obtained from $\textbf{c}^*_{a,b,c,d}$ by deleting the first entry. 
\end{corollary}

\section{Subfield codes from specific functions}
After having presented general results for the codes and subcodes involved in our framework, in this section we will specify some functions $f,g,h:\mathbb{F}_{q^m}\rightarrow\mathbb{F}_q$ in order to determine the values of the terms $\Phi_{f,g,h}$ and $\Lambda_{a,b,c,d}^{f,g,h}$.

\subsection{$f(x)=\text{Tr}_{q^m/q}(x)$, $g(y)=\text{Tr}_{q^m/q}(y^2)$, $h(z)=\text{Norm}_{q^m/q}(z)$}

Firstly, we consider the following three functions in order to build our codes $\mathcal{C}_{f,g,h}$, $\mathcal{C}_{f,g,h}^{(q)}$ and their punctured versions: $f(x)=\text{Tr}_{q^m/q}(x)$, $g(y)=\text{Tr}_{q^m/q}(y^2)$, $h(z)=\text{Norm}_{q^m/q}(z)$. We recall that the \emph{norm} of an element $z \in \mathbb{F}_{q^m}$ with respect to $\mathbb{F}_q$ is the element of $\mathbb{F}_q$ defined by
\[
\mathrm{Norm}_{{q^m}/q}(z) = z \cdot z^q \cdot z^{q^2} \cdots v^{z^{m-1}} = \prod_{i=0}^{m-1} z^{q^i}.
\]
Using the general results of the previous section, we determine the parameters of the codes.
\begin{lemma}\label{lemma: length A}
    For the latter choice of functions, the lengths of the codes $\mathcal{C}_{f,g,h}$ and $\mathcal{C}_{f,g,h}^{(q)}$ is $n=1+q^{3m-1}$. In addition, the dimension of the code $\mathcal{C}_{f,g,h}$ is equal to $4$. 
\end{lemma}
\begin{proof}
    From Lemma \ref{lemma: length}, we have that
    $$n=1+q^{3m-1}+\frac{1}{q}\Phi_{f,g,h},$$
    hence we have to determine the value of the parameter $\Phi_{f,g,h}$.
$$\Phi_{f,g,h}=\sum_{s\in\mathbb{F}_q^*}\sum_{x\in\mathbb{F}_{q^m}}\chi(s\text{Tr}_{q^m/q}(x))\sum_{y\in\mathbb{F}_{q^m}}\chi(s\text{Tr}_{q^m/q}(y^2))\sum_{z\in\mathbb{F}_{q^m}}\chi(s\text{Norm}_{q^m/q}(z))=0,$$
since, for $s\in\mathbb{F}_q^*$, by orthogonality of characters:
$$\sum_{x\in\mathbb{F}_{q^m}}\chi(s\text{Tr}_{q^m/q}(x))=\sum_{x\in\mathbb{F}_{q^m}}\chi'(sx)=0.$$
In conclusion, $n=1+q^{3m-1}+\frac{1}{q}\Phi_{f,g,h}=1+q^{3m-1}$.
\end{proof}
After having determined the value of the parameter $\Phi_{f,g,h}$, we compute the other $\Lambda_{a,b,c,d}^{f,g,h}$, for $(a,b,c,d)\in\mathbb{F}_q\times\mathbb{F}_{q^m}^3$. We start with an even $q$.
\begin{lemma}\label{lemma: Lambda A q even}
    With the latter choice of functions and $q$ even, for $(a,b,c,d)\in\mathbb{F}_q\times\mathbb{F}_{q^m}^3$ we have 
    \begin{equation}
        \Lambda_{a,b,c,d}^{f,g,h}=\begin{cases}
            0 \,\,\,\,\,\,\,\,\,\,\,\,\,\,\,\,\,\,\,\,\,\,\,\,\,\,\,\,\,\,\,\,\,\,\,\,\,\,\,\,\,\,\,\,\,\,\,\,\,\,\,\,\,\,\,\,\,\,\,\,\,\,\,\,\,\,\,\,\,\,\,\,\,\,\,\,\,\,\,\,\,\,\,\,\,\,\,\,\,\,\,\,\,\,\,\,\,\,\,\,\,\,\,\,\,\,\,\,\,\,\,\,\,\,\,\,\,\,\,\,\,\,\,\,\,\,\,\,\,\,\,\,\,\,\,\,\,\text{if} \,\,b\in \mathbb{F}_{q^m}\setminus\mathbb{F}_{q}^*\,\,\text{or}\,\,c\in \mathbb{F}_{q^m}\setminus\mathbb{F}_{q}^*\\
            q^{2m}\sum_{z\in\mathbb{F}_{q^m}}\chi(\frac{b}{c^2}(a+\text{Tr}_{q^m/q}(dz)-b\text{Norm}_{q^m/q}(z)))\,\,\,\,\,\text{otherwise}.
        \end{cases}\notag
    \end{equation}
\end{lemma}
\begin{proof}
    For general values of $(a,b,c,d)\in\mathbb{F}_q\times\mathbb{F}_{q^m}^3$:
    \begin{align}
\Lambda_{a,b,c,d}^{f,g,h}&=\sum_{s\in\mathbb{F}_q^*}\chi(sa)\sum_{\omega\in\mathbb{F}_q^*}\sum_{(x,y,z)\in\mathbb{F}_{q^m}^3}\chi(\omega\text{Tr}_{q^m/q}(x)+\omega \text{Tr}_{q^m/q}(y^2)+\omega \text{Norm}_{q^m/q}(z))\chi'(sbx+scy+sdz)\notag\\
        &=\sum_{s\in\mathbb{F}_q^*}\chi(sa)\sum_{\omega\in\mathbb{F}_q^*}\sum_{x\in\mathbb{F}_{q^m}}\chi'((\omega+sb)x)\sum_{y\in\mathbb{F}_{q^m}}\chi'(\omega y^2+scy)\sum_{z\in\mathbb{F}_{q^m}} \chi(\omega\text{Norm}_{q^m/q}(z)+\text{Tr}_{q^m/q}(sdz))\notag
    \end{align}
    If $b\in \mathbb{F}_{q^m}\setminus\mathbb{F}_{q}^*$, then $\omega+sb\not =0$, which means that $\sum_{w\in\mathbb{F}_{q^m}}\chi'((\omega+sb)x)=0$ and hence $\Lambda_{a,b,c,d}^{f,g,h}=0$.
    \\\\
    If $b\in\mathbb{F}_q^*$:
    $$\Lambda_{a,b,c,d}^{f,g,h}=q^m\sum_{\substack{(s,\omega)\in\mathbb{F}_q^{*2}\\\omega+sb=0}}\sum_{(y,z)\in\mathbb{F}_{q^m}^2}\chi'(-sby^2+scy)\sum_{z\in\mathbb{F}_{q^m}}\chi(s(-b\text{Norm}_{q^m/q}(z)+\text{Tr}_{q^m/q}(dz)+a)).$$
    Since $q$ is even, if $c\in \mathbb{F}_{q^m}\setminus\mathbb{F}_{q}^*$, then $b\not =sc^2$ and hence $\Lambda_{a,b,c,d}^{f,g,h}=0$.
    In the remaining cases:
    \begin{align}
        \Lambda_{a,b,c,d}^{f,g,h}&=q^{2m}\sum_{\substack{s\in\mathbb{F}_q^*\\b=sc^2}}\sum_{z\in\mathbb{F}_{q^m}}\chi(s(a+\text{Tr}_{q^m/q}(dz)-b\text{Norm}_{q^m/q}(z))\notag\\
        &=q^{2m}\sum_{z\in\mathbb{F}_{q^m}}\chi(\frac{b}{c^2}(a+\text{Tr}(dz)-b\text{Norm}(z))),\notag
    \end{align}
    as we wanted to prove.
\end{proof}
In the binary case, that is for $q=2$, the quantity $\Lambda_{a,b,c,d}^{f,g,h}$ is non-zero only when $b = c = 1$, and in this situation we have
    $$\Lambda_{a,1,1,d}^{f,g,h}=2^{2m}\sum_{z\in\mathbb{F}_{2^m}}\chi(a+\text{Tr}_{2^m/2}(dz)+\text{Norm}_{2^m/2}(z)).$$
In addition, for $m=2$, we obtain the following families of codes.
\begin{prop}\label{prop: code over F4}
    Let $f(x)=\text{Tr}_{4/2}(x)$, $g(y)=\text{Tr}_{4/2}(y^2)$, $h(z)=\text{Norm}_{4/2}(z)$. Then $\mathcal{C}_{f,g,h}$ is a $[33,4]$ linear code over $\mathbb{F}_4$. The subfield code $\mathcal{C}_{f,g,h}^{(2)}$ is a $[33,7,8]$ binary linear code with the following weight distribution
    \begin{center}
\begin{tabular}{c|c}
\hline Weight & Multiplicity \\
\hline 0 & 1 \\
$32$ & $1$ \\
$24$ & $4$ \\
$17$ & $64$ \\
$16$ & $54$ \\
$8$ & $4$\\
\hline
\end{tabular}
\end{center} 
The dual code $\mathcal{C}_{f,g,h}^{(2)\perp}$ is a $[33,26,3]$ binary linear code  almost dimensionally optimal with respect to the Griesmer bound.
The subfield code of the punctured code $\mathcal{C}_{f,g,h}^{*(2)}$ is a $[32,7,8]$ binary linear code with the following weight distribution
    \begin{center}
\begin{tabular}{c|c}
\hline Weight & Multiplicity \\
\hline 0 & 1 \\
$32$ & $1$ \\
$24$ & $4$ \\
$16$ & $118$ \\
$8$ & $4$\\
\hline
\end{tabular}
\end{center} 
The dual code $\mathcal{C}_{f,g,h}^{*(2)\perp}$ is a $[32,25,4]$ binary linear code  distance optimal with respect to the Sphere Packing bound.
\end{prop}
\begin{proof}
    The parameters of the code are determined using Lemma \ref{lemma: length A} in the specific case $q=m=2$. It remains to compute the weights of the subfield code via Lemma \ref{lemma: weights}. For $\textbf{c}_{a,b,c,d}$, a codeword in the code parameterized by $a\in\mathbb{F}_2$ and $b,c,d\in\mathbb{F}_4$, we have
    \begin{align}\notag
        \text{wt}(\textbf{c}_{a,b,c,d})&=\begin{cases}\notag
            0 \,\,\,\,\,\,\,\,\,\,\,\,\,\,\,\,\,\,\,\,\,\,\,\,\,\,\,\,\,\,\,\,\,\,\,\,\,\,\,\,\,\,\,\,\,\,\,\,\,\,\,\,\,\,\text{if} \,\,\,a=b=c=d=0\\
            32 \,\,\,\,\,\,\,\,\,\,\,\,\,\,\,\,\,\,\,\,\,\,\,\,\,\,\,\,\,\,\,\,\,\,\,\,\,\,\,\,\,\,\,\,\,\,\,\,\,\,\,\text{if} \,\,\,a\not=0, b=c=d=0\\
           \delta(b)+16-\frac{1}{4}\Lambda_{a,b,c,d}^{f,g,h} \,\,\,\,\text{if} \,\,\,b=c=1\\
            \delta(b)+16 \,\,\,\,\,\,\,\,\,\,\,\,\,\,\,\,\,\,\,\,\,\,\,\,\,\,\,\,\,\,\,\,\,\text{otherwise}      
            \end{cases}
        \end{align}
        Hence, by computing the explicit values of $\Lambda_{a,b,c,d}^{f,g,h}$ in this case over $\mathbb{F}_4$, we obtain the weights of the code and, by Corollary \ref{corollary: parameters of the punctured codes}, the weights of the subfield code of the punctured code. The distances of the duals can be computed via the Pless power moments. 
\end{proof}
Finally, for $q$ odd and $m$ even, the value of the parameter $\Lambda_{a,b,c,d}^{f,g,h}$ is the following.
\begin{lemma}\label{lemma: Lambda A q odd m even}
     For the latter choice of functions, $q$ odd and $m$ even, for $(a,b,c,d)\in\mathbb{F}_q\times\mathbb{F}_{q^m}^3$ we have
     $$\Lambda_{a,b,c,d}^{f,g,h}=G(\eta',\chi')q^m\sum_{s\in\mathbb{F}_q^*}\sum_{z\in\mathbb{F}_{q^m}}\chi(s(a+\text{Tr}_{q^m/q}(\frac{c^2}{4b})+\text{Tr}_{q^m/q}(dz)-b\text{Norm}_{q^m/q}(z))).$$
\end{lemma}
\begin{proof}
We apply Lemma \ref{Character and quadratic polynomial} to evaluate the Weil sum $\sum_{y\in\mathbb{F}_{q^m}}\chi'(-sby^2+scy)$. This yields the following computation:
    \begin{align}
        \Lambda_{a,b,c,d}^{f,g,h}&=q^m\sum_{s\in\mathbb{F}_q^*}\chi(sa)\sum_{y\in\mathbb{F}_{q^m}}\chi'(-sby^2+scy)\sum_{z\in\mathbb{F}_{q^m}}\chi(s(\text{Tr}_{q^m/q}(dz)-b\text{Norm}_{q^m/q}(z))\notag\\
        &=q^m\sum_{s\in\mathbb{F}_q^*}\chi'(\frac{sc^2}{4b})G(\eta',\chi')\sum_{z\in\mathbb{F}_{q^m}}\chi(s(a+\text{Tr}_{q^m/q}(dz)-b\text{Norm}_{q^m/q}(z)))\notag\\
        &=G(\eta',\chi')q^m\sum_{s\in\mathbb{F}_q^*}\sum_{z\in\mathbb{F}_{q^m}}\chi(s(a+\text{Tr}_{q^m/q}(\frac{c^2}{4b})+\text{Tr}_{q^m/q}(dz)-b\text{Norm}_{q^m/q}(z))).\nonumber
    \end{align}
\end{proof}

\subsection{$f(x)=\text{Tr}_{q^m/q}(x)$, $g(y)=\text{Tr}_{q^m/q}(y^2)$, $h(z)$ generic}

Secondly, we consider the following three functions in order to build our codes $\mathcal{C}_{f,g,h}$, $\mathcal{C}_{f,g,h}^{(q)}$ and their punctured versions: $f(x)=\text{Tr}_{q^m/q}(x)$, $g(y)=\text{Tr}_{q^m/q}(y^2)$ and $h(z)$ is a generic function from $\mathbb{F}_{q^m}$ to $\mathbb{F}_q$. Using the general results from the previous section, we determine the code parameters.
\begin{lemma}\label{lemma: length B}
    For the latter choice of functions, the lengths of the codes $\mathcal{C}_{f,g,h}$ and $\mathcal{C}_{f,g,h}^{(q)}$ are $n=1+q^{3m-1}$. In addition, the dimension of the code $\mathcal{C}_{f,g,h}$ is equal to $4$. 
\end{lemma}
\begin{proof}
    From Lemma \ref{lemma: length}, we have that
    $$n=1+q^{3m-1}+\frac{1}{q}\Phi_{f,g,h},$$
    hence we have to determine the value of the parameter $\Phi_{f,g,h}$.
$$\Phi_{f,g,h}=\sum_{s\in\mathbb{F}_q^*}\sum_{x\in\mathbb{F}_{q^m}}\chi(s\text{Tr}_{q^m/q}(x))\sum_{y\in\mathbb{F}_{q^m}}\chi(s\text{Tr}_{q^m/q}(y^2))\sum_{z\in\mathbb{F}_{q^m}}\chi(sh(z))=0,$$
since, for $s\in\mathbb{F}_q^*$, by orthogonality of characters:
$$\sum_{x\in\mathbb{F}_{q^m}}\chi(s\text{Tr}_{q^m/q}(x))=\sum_{x\in\mathbb{F}_{q^m}}\chi'(sx)=0.$$
In conclusion, $n=1+q^{3m-1}+\frac{1}{q}\Phi_{f,g,h}=1+q^{3m-1}$.
\end{proof}
After having determined the value of the parameter $\Phi_{f,g,h}$, we compute the other $\Lambda_{a,b,c,d}^{f,g,h}$, for $(a,b,c,d)\in\mathbb{F}_q\times\mathbb{F}_{q^m}^3$ in the binary case, that is, for $q=2$.
\begin{lemma}\label{lemma: Lambda B q=2}
    With the latter choice of functions and $q=2$, for $(a,b,c,d)\in\mathbb{F}_2\times\mathbb{F}_{2^m}^3$ we have 
    \begin{equation}
        \Lambda_{a,b,c,d}^{f,g,h}=\begin{cases}
            (-1)^a2^{2m}\mathcal{W}_h(d) \,\,\,\,\text{if} \,\,b=1\,\,\text{and}\,\,c=1\\
            0\,\,\,\,\,\,\,\,\,\,\,\,\,\,\,\,\,\,\,\,\,\,\,\,\,\,\,\,\,\,\,\,\,\,\,\,\,\,\,\,\,\,\,\text{otherwise},
        \end{cases}\notag
    \end{equation}
    where $\mathcal{W}_h$ is the Walsh transform of the function $h$.
\end{lemma}
\begin{proof}
    For general values of $(a,b,c,d)\in\mathbb{F}_2\times\mathbb{F}_{2^m}^3$:
    \begin{align}
\Lambda_{a,b,c,d}^{f,g,h}&=\sum_{s\in\mathbb{F}_2^*}\chi(sa)\sum_{\omega\in\mathbb{F}_2^*}\sum_{(x,y,z)\in\mathbb{F}_{2^m}^3}\chi(\omega\text{Tr}_{q^m/q}(x)+\omega \text{Tr}_{q^m/q}(y^2)+\omega B(z))\chi'(sbx+scy+sdz)\notag\\
        &=(-1)^a\sum_{(x,y,z)\in\mathbb{F}_{2^m}^3}\chi(\text{Tr}_{q^m/q}(x)+ \text{Tr}_{q^m/q}(y^2)+ B(z))\chi'(sbx+scy+sdz)\notag\\
        &=(-1)^a\sum_{x\in\mathbb{F}_{2^m}}\chi'((1+b)x)\sum_{y\in\mathbb{F}_{2^m}}\chi'(y^2+cy)\sum_{z\in\mathbb{F}_{2^m}}\chi(B(z)+\text{Tr}_{q^m/q}(dz))\notag
    \end{align}
    If $b\not=1$, then $\sum_{x\in\mathbb{F}_{2^m}}\chi'((1+b)x)=0$ and hence $\Lambda_{a,b,c,d}^{f,g,h}=0$.
    \\\\
    If $b=1$:
$$\Lambda_{a,b,c,d}^{f,g,h}=2^m(-1)^a\sum_{y\in\mathbb{F}_{2^m}}\chi'(y^2+cy)\sum_{z\in\mathbb{F}_{2^m}}\chi(B(z)+\text{Tr}_{q^m/q}(dz)).$$
Now, if $c\not=1$, then $\Lambda_{a,b,c,d}^{f,g,h}=0$.
\\\\
If $c=1$:
$$\Lambda_{a,b,c,d}^{f,g,h}=(-1)^a2^{2m}\mathcal{W}_h(d),$$
as we wanted to prove.
\end{proof}
    Notice that for specific families of functions, such as s-plateaued, the values of the Walsh transform of $h$ are known; hence it is possible to compute all the parameters and weights of the subfield codes $\mathcal{C}_{f,g,h}^{(q)}$ and their punctured versions. Even the parameters of the dual codes can be determined. 
    We start by analyzing the bent case, hence $m$ must be even.
\begin{thm}\label{thm: parameters in the bent case}
    Let $m$ be even, $f(x)=\text{Tr}_{2^m/2}(x)$, $g(y)=\text{Tr}_{2^m/2}(y^2)$ and $h(z)=B(z)$, a bent function from $\mathbb{F}_{2^m}$ to $\mathbb{F}_2$. Then $\mathcal{C}_{f,g,h}^{(2)}$ is a five-weights $[2^{3m-1}+1, 3m+1, 2^{3m-2}-2^\frac{5m-4}{2}]$ binary linear code, with the following weight distribution:
\begin{center}
\begin{tabular}{c|c}
\hline Weight & Multiplicity \\
\hline 0 & 1 \\
$2^{3m-2}-2^{\frac{5m-4}{2}}$ & $2^m$ \\
$2^{3m-2}$ & $2^{3 m}-2^{m+1}-2$ \\
$2^{3m-2}+1$ & $2^{3 m}$ \\
$2^{3m-2}+2^{\frac{5m-4}{2}}$ & $2^m$ \\
$2^{3m-1}$ & 1 \\
\hline
\end{tabular}
\end{center}
The dual code $\mathcal{C}_{f,g,h}^{(q)\perp}$ is a $[2^{3m-1}+1, 2^{3m-1}-3m,3]$ binary linear code that is almost dimensionally optimal with respect to the sphere packing bound.
\end{thm}
\begin{proof}
    From Lemma \ref{lemma: length B}, we obtain that the length of the code is equal to $n=2^{3m-1}+1$ and, combining the results of Lemma \ref{lemma: length} and Lemma \ref{lemma: Lambda B q=2}, we can compute the weights of the codewords in the following way. For $(a,b,c,d)\in\mathbb{F}_2\times \mathbb{F}_{2^m}^3$:
    \begin{align*}
        \text{wt}(\textbf{c}_{a,b,c,d})&=\begin{cases}\notag
            0 \,\,\,\,\,\,\,\,\,\,\,\,\,\,\,\,\,\,\,\,\,\,\,\,\,\,\,\,\,\,\,\,\,\,\,\,\,\,\,\,\,\,\text{if} \,\,\,a=b=c=d=0\\
            2^{3m-1} \,\,\,\,\,\,\,\,\,\,\,\,\,\,\,\,\,\,\,\,\,\,\,\,\,\,\,\,\,\text{if} \,\,\,a\not=0, b=c=d=0\\
            2^{3m-2}-2^{\frac{5m-4}{2}} \,\,\,\,\,\,\text{if} \,\,\,b=c=1, (-1)^aW_B(d)=2^\frac{m}{2}\\
            2^{3m-2}+2^{\frac{5m-4}{2}} \,\,\,\,\,\,\text{if}\,\,\, b=c=1, (-1)^aW_B(d)=-2^\frac{m}{2}\\
            2^{3m-2} \,\,\,\,\,\,\,\,\,\,\,\,\,\,\,\,\,\,\,\,\,\,\,\,\,\,\,\,\,\,\text{if}\,\,\,\delta(b)=0, (b,c)\not=(1,1),\,\,\, b,c,d \,\,\,\text{not all 0}\\
            1+2^{3m-2}\,\,\,\,\,\,\,\,\,\,\,\,\,\,\,\,\,\,\,\, \text{if}\,\,\,\delta(b)=1        
            \end{cases}\\
        &=\begin{cases}
            0 \,\,\,\,\,\,\,\,\,\,\,\,\,\,\,\,\,\,\,\,\,\,\,\,\,\,\,\,\,\,\,\,\,\,\,\,\,\,\,\,\,\,\text{with 1 time}\\
            2^{3m-1} \,\,\,\,\,\,\,\,\,\,\,\,\,\,\,\,\,\,\,\,\,\,\,\,\,\,\,\,\,\text{with 1 time}\\
            2^{3m-2}-2^{\frac{5m-4}{2}} \,\,\,\,\,\,\text{with} \,\,\,2^m\,\,\,\text{times}\\
            2^{3m-2}+2^{\frac{5m-4}{2}} \,\,\,\,\,\,\text{with} \,\,\,2^m\,\,\,\text{times}\\
            2^{3m-2} \,\,\,\,\,\,\,\,\,\,\,\,\,\,\,\,\,\,\,\,\,\,\,\,\,\,\,\,\,\text{with}\,\,\,2^{3m}-2^{m+1}-2\,\,\,\text{times}\\
            1+2^{3m-2} \,\,\,\,\,\,\,\,\,\,\,\,\,\,\,\,\,\,\,\text{with}\,\,\,2^{3m}\,\,\,\text{times},
        \end{cases}
    \end{align*}
    which yields the code's weight distribution.\\\\
    For the dual code, let $d^\perp$ be its minimum distance. From the sphere packing bound, we have
    $$2^{3m+1}\ge \sum_{i=0}^{\left\lfloor{\frac{d^\perp-1}{2}}\right\rfloor}\binom{2^{3m-1}+1}{i}.$$
    Hence, we can deduce that $d^\perp\le 4$ and from the first four Pless power moments we conclude that actually $d^\perp=3$.
\end{proof}
By applying Corollary \ref{corollary: parameters of the punctured codes} to Theorem \ref{thm: parameters in the bent case}, we can deduce the parameters of the punctured code and its dual. 
\begin{corollary}\label{corollary: parameters of the punctured code in the bent case}
    Let $m$ be even, $f(x)=\text{Tr}_{2^m/2}(x)$, $g(y)=\text{Tr}_{2^m/2}(y^2)$ and $h(z)=B(z)$, a bent function from $\mathbb{F}_{2^m}$ to $\mathbb{F}_2$. Then $\mathcal{C}_{f,g,h}^{*(2)}$ is a four-weights $[2^{3m-1}, 3m+1, 2^{3m-2}-2^\frac{5m-4}{2}]$ binary linear code, with the following weight distribution
\begin{center}
\begin{tabular}{c|c}
\hline Weight & Multiplicity \\
\hline 0 & 1 \\
$2^{3m-2}-2^{\frac{5m-4}{2}}$ & $2^m$ \\
$2^{3m-2}$ & $2^{3 m+1}-2^{m+1}-2$ \\
$2^{3m-2}+2^{\frac{5m-4}{2}}$ & $2^m$ \\
$2^{3m-1}$ & 1 \\
\hline
\end{tabular}
\end{center}
The dual code $\mathcal{C}_{f,g,h}^{*(q)\perp}$ is a $[2^{3m-1}, 2^{3m-1}-3m-1,4]$ binary linear code, that is almost dimensionally optimal with respect to the sphere packing bound. 
\end{corollary}


More generally, we analyze the non-bent $s$-plateaued case, that is, $s>0$. We begin by assuming that $m$ is even; consequently, $s$ is even as well.

\begin{thm}\label{thm: parameters in the boolean plateaued case, m even}
    Let $m$ be even, $f(x)=\text{Tr}_{2^m/2}(x)$, $g(y)=\text{Tr}_{2^m/2}(y^2)$ and $h(z)=P(z)$, a normalized $s$-plateaued function from $\mathbb{F}_{2^m}$ to $\mathbb{F}_2$, that is, $P(0)=0$. Then $\mathcal{C}_{f,g,h}^{(2)}$ is a five-weights $[2^{3m-1}+1, 3m+1, 2^{3m-2}-2^\frac{5m+s-4}{2}]$ binary linear code, with the following weight distribution:
\begin{center}
\begin{tabular}{c|c}
\hline Weight & Multiplicity \\
\hline 0 & 1 \\
$2^{3m-2}-2^{\frac{5m+s-4}{2}}$ & $2^{m-s}$ \\
$2^{3m-2}$ & $2^{3 m}-2^{m-s+1}-2$ \\
$2^{3m-2}+1$ & $2^{3 m}$ \\
$2^{3m-2}+2^{\frac{5m+s-4}{2}}$ & $2^{m-s}$ \\
$2^{3m-1}$ & 1 \\
\hline
\end{tabular}
\end{center}
The dual code $\mathcal{C}_{f,g,h}^{(q)\perp}$ is a $[2^{3m-1}+1, 2^{3m-1}-3m,3]$ binary linear code that is almost dimensionally optimal with respect to the sphere packing bound.
The puncture code $\mathcal{C}_{f,g,h}^{*(2)}$ is a four-weights $[2^{3m-1}, 3m+1, 2^{3m-2}-2^\frac{5m+s-4}{2}]$ binary linear code, with the following weight distribution
\begin{center}
\begin{tabular}{c|c}
\hline Weight & Multiplicity \\
\hline 0 & 1 \\
$2^{3m-2}-2^{\frac{5m+s-4}{2}}$ & $2^{m-s}$ \\
$2^{3m-2}$ & $2^{3 m+1}-2^{m-s+1}-2$ \\
$2^{3m-2}+2^{\frac{5m+s-4}{2}}$ & $2^{m-s}$ \\
$2^{3m-1}$ & 1 \\
\hline
\end{tabular}
\end{center}
The dual code $\mathcal{C}_{f,g,h}^{*(q)\perp}$ is a $[2^{3m-1}, 2^{3m-1}-3m-1,4]$ binary linear code, that is almost dimensionally optimal with respect to the sphere packing bound.
\end{thm}
\begin{proof}
    As in the proof of Theorem \ref{thm: parameters in the bent case}, from Lemma \ref{lemma: length B} we obtain the length of the code and, combining the results of Lemma \ref{lemma: length} and Lemma \ref{lemma: Lambda B q=2}, we compute the weights of the codewords. For $(a,b,c,d)\in\mathbb{F}_2\times \mathbb{F}_{2^m}^3$:
    \begin{align*}
        \text{wt}(\textbf{c}_{a,b,c,d})&=\begin{cases}\notag
            0 \,\,\,\,\,\,\,\,\,\,\,\,\,\,\,\,\,\,\,\,\,\,\,\,\,\,\,\,\,\,\,\,\,\,\,\,\,\,\,\,\,\,\,\,\,\,\,\text{if} \,\,\,a=b=c=d=0\\
            2^{3m-1} \,\,\,\,\,\,\,\,\,\,\,\,\,\,\,\,\,\,\,\,\,\,\,\,\,\,\,\,\,\,\,\,\,\,\text{if} \,\,\,a\not=0, b=c=d=0\\
            2^{3m-2}-2^{\frac{5m+s-4}{2}} \,\,\,\,\,\,\text{if} \,\,\,b=c=1, (-1)^aW_P(d)=2^\frac{m+s}{2}\\
            2^{3m-2}+2^{\frac{5m+s-4}{2}} \,\,\,\,\,\,\text{if}\,\,\, b=c=1, (-1)^aW_P(d)=-2^\frac{m+s}{2}\\
            2^{3m-2} \,\,\,\,\,\,\,\,\,\,\,\,\,\,\,\,\,\,\,\,\,\,\,\,\,\,\,\,\,\,\,\,\,\,\,\text{if}\,\,\,\delta(b)=0, (b,c)\not=(1,1),\,\,\, b,c,d \,\,\,\text{not all 0}\\
            1+2^{3m-2}\,\,\,\,\,\,\,\,\,\,\,\,\,\,\,\,\,\,\,\,\,\,\,\,\, \text{if}\,\,\,\delta(b)=1        
            \end{cases}\\
        &=\begin{cases}
            0 \,\,\,\,\,\,\,\,\,\,\,\,\,\,\,\,\,\,\,\,\,\,\,\,\,\,\,\,\,\,\,\,\,\,\,\,\,\,\,\,\,\,\text{with 1 time}\\
            2^{3m-1} \,\,\,\,\,\,\,\,\,\,\,\,\,\,\,\,\,\,\,\,\,\,\,\,\,\,\,\,\,\text{with 1 time}\\
            2^{3m-2}-2^{\frac{5m-4}{2}} \,\,\,\,\,\,\text{with} \,\,\,2^{m-s}\,\,\,\text{times}\\
            2^{3m-2}+2^{\frac{5m-4}{2}} \,\,\,\,\,\,\text{with} \,\,\,2^{m-s}\,\,\,\text{times}\\
            2^{3m-2} \,\,\,\,\,\,\,\,\,\,\,\,\,\,\,\,\,\,\,\,\,\,\,\,\,\,\,\,\,\text{with}\,\,\,2^{3m}-2^{m-s+1}-2\,\,\,\text{times}\\
            1+2^{3m-2} \,\,\,\,\,\,\,\,\,\,\,\,\,\,\,\,\,\,\,\text{with}\,\,\,2^{3m}\,\,\,\text{times},
        \end{cases}
    \end{align*}
    which yields the code's weight distribution. Applying Corollary \ref{corollary: parameters of the punctured codes}, we also determine the parameters and the weight distribution of the punctured code.
    For the dual codes, we compute their minimum distances from the first four Pless power moments.
\end{proof}
Finally, we assume that $m$ is odd; hence $s$ is also odd. For instance, this includes the almost bent case $s=1$. The proof proceeds analogously to that of Theorem~\ref{thm: parameters in the boolean plateaued case, m even}, and is therefore omitted.

\begin{thm}\label{thm: parameters in the boolean plateaued case, m odd}
    Let $m$ be odd, $f(x)=\text{Tr}_{2^m/2}(x)$, $g(y)=\text{Tr}_{2^m/2}(y^2)$ and $h(z)=P(z)$, a normalized $s$-plateaued function from $\mathbb{F}_{2^m}$ to $\mathbb{F}_2$, that is, $P(0)=0$. Then $\mathcal{C}_{f,g,h}^{(2)}$ is a five-weights $[2^{3m-1}+1, 3m+1, 1+2^{3m-2}-2^\frac{5m+s-4}{2}]$ binary linear code, with the following weight distribution:
\begin{center}
\begin{tabular}{c|c}
\hline Weight & Multiplicity \\
\hline 0 & 1 \\
$1+2^{3m-2}-2^{\frac{5m+s-4}{2}}$ & $2^{m-s}$ \\
$2^{3m-2}$ & $2^{3 m}-2$ \\
$2^{3m-2}+1$ & $2^{3 m}-2^{m-s+1}$ \\
$1+2^{3m-2}+2^{\frac{5m+s-4}{2}}$ & $2^{m-s}$ \\
$2^{3m-1}$ & 1 \\
\hline
\end{tabular}
\end{center}
The dual code $\mathcal{C}_{f,g,h}^{(q)\perp}$ is a $[2^{3m-1}+1, 2^{3m-1}-3m,4]$ binary linear code that is distance and dimensionally optimal with respect to the sphere packing bound.
The puncture code $\mathcal{C}_{f,g,h}^{*(2)}$ is a four-weights $[2^{3m-1}, 3m+1, 2^{3m-2}-2^\frac{5m+s-4}{2}]$ binary linear code, with the following weight distribution
\begin{center}
\begin{tabular}{c|c}
\hline Weight & Multiplicity \\
\hline 0 & 1 \\
$2^{3m-2}-2^{\frac{5m+s-4}{2}}$ & $2^{m-s}$ \\
$2^{3m-2}$ & $2^{3 m+1}-2^{m-s+1}-2$ \\
$2^{3m-2}+2^{\frac{5m+s-4}{2}}$ & $2^{m-s}$ \\
$2^{3m-1}$ & 1 \\
\hline
\end{tabular}
\end{center}
The dual code $\mathcal{C}_{f,g,h}^{*(q)\perp}$ is a $[2^{3m-1}, 2^{3m-1}-3m-1,4]$ binary linear code, that is almost dimensionally optimal with respect to the sphere packing bound.
\end{thm}
\begin{remark}
   Notice that the punctured codes of Theorem~\ref{thm: parameters in the boolean plateaued case, m even} have the same parameters and weight distributions as the punctured codes of Theorem~\ref{thm: parameters in the boolean plateaued case, m odd}, with the subtle difference that in the first case $s$ must be even, whereas in the second case it must be odd.
    \end{remark}

\section{The vectorial framework}
In this section, we present the adaptation of the code construction from \cite{xu} to vectorial functions. Let $f$ and $g$ be two functions from $\mathbb{F}_{q^m}$ to $\mathbb{F}_{q^m}$. We define the subset $\mathcal{D}$ of $\mathbb{F}_{q^m}^2$ as
\begin{equation}\label{vectorial defining set}
    \mathcal{D}=\left\{(x, y) \in \mathbb{F}_{q^m}^2: f(x)+g(y)=0\right\}
\end{equation}
and a $3 \times(\# \mathcal{D})$ matrix $G_{f, g}^*$ over $\mathbb{F}_{q^m}$ as 
\begin{equation}\label{vectorial punctured matrix}
    G_{f, g}^*=\left(\begin{array}{l} 1 \\ x \\ y \end{array}\right)_{(x, y) \in \mathcal{D}}.
\end{equation}
Let $\mathcal{C}_{f,g}$ be the $[\# \mathcal{D}+1,3]$ linear code over $\mathbb{F}_{q^m}$ generated by the following matrix:
\begin{equation}\label{vectorial matrix}
    G_{f, g}=\left(\begin{array}{cc} 0 & \\ 1 & G_{f, g}^* \\ 0 & \end{array}\right).
\end{equation}
Let $\mathcal{C}_{f,g}^*$ be the punctured code obtained from $\mathcal{C}_{f,g}$ by puncturing on the first coordinate, hence $\mathcal{C}_{f,g}^*$ is the $[\# \mathcal{D},3]$ linear code over $\mathbb{F}_{q^m}$ with generator matrix $G_{f, g}^*$.\\\\
The aim of this section is to determine the parameters of the codes and their subfield codes, with specific attention to the binary case.
\\\\
In the vectorial framework, the trace representation of the subfield code $\mathcal{C}_{f, g}^{(q)}$ is given by $$ \mathcal{C}_{f, g}^{(q)}=\left\{\mathbf{c}_{a, b, c}: a \in \mathbb{F}_q, b, c \in \mathbb{F}_{q^m}\right\}, $$ where $$ \mathbf{c}_{a, b, c}=\left(\operatorname{Tr}_{q^m / q}(b),\left(a+\operatorname{Tr}_{q^m / q}(b x+c y)\right)_{(x, y) \in \mathcal{D}}\right). $$
As before, let $\eta$ and $\eta '$ be the quadratic
characters of $\mathbb{F}_q$ and $\mathbb{F}_{q^m}$, respectively. Also, let $\chi$ and $\chi '$ be the canonical additive characters of $\mathbb{F}_q$ and $\mathbb{F}_{q^m}$, respectively. Hence, $\chi'=\chi\circ\text{Tr}_{q^m/q}$. Consider two functions $f$ and $g$ from $\mathbb{F}_{q^m}$ to $\mathbb{F}_{q^m}$.\\
In this case, we will need the following two terms depending on the functions.
\begin{equation}\label{Phi vectorial}
\overline{\Phi}_{f,g}=\sum_{s\in\mathbb{F}_{q^m}^*}\sum_{x\in\mathbb{F}_{q^m}}\chi'(sf(x))\sum_{y\in\mathbb{F}_{q^m}}\chi'(sg(y)).
\end{equation}
Notice that if at least one function among $f$ and $g$ is a permutation polynomial, we have $\overline{\Phi}_{f,g}=0$ by the orthogonality relation for additive
characters.\\\\
Then, for $(a,b,c)\in\mathbb{F}_q\times\mathbb{F}_{q^m}^2$:
\begin{equation}\label{Vectorial lambda}
    \overline{\Lambda}_{a,b,c}^{f,g}=\sum_{s\in\mathbb{F}_q^*}\chi(sa)\sum_{\omega\in\mathbb{F}_{q^m}^*}\sum_{(x,y)\in\mathbb{F}_{q^m}^2}\chi'(\omega f(x)+\omega g(y))\chi'(sbx+scy).
\end{equation}

In the following lemmas, we will compute the length and weights of the codes with respect to the latter terms.

\begin{lemma}\label{lemma: vectorial length}
    The length of the codes $\mathcal{C}_{f,g}$ and $\mathcal{C}_{f,g}^{(q)}$ is $n=1+q^{m}+\frac{1}{q^m}\overline{\Phi}_{f,g}$.
\end{lemma}
\begin{proof}
    Let $n$ be the length of the codes. From construction \eqref{vectorial matrix} it is clear that $n=\# \mathcal{D}+1$, due to the first column of the generator matrix. Hence, we compute the cardinality of the defining set $\mathcal{D}$ using the well-known results of orthogonality in character theory.
    \begin{align}
        \# \mathcal{D}&=\frac{1}{q^m}\sum_{(x,y)\in\mathbb{F}_{q^m}^2}\sum_{s\in\mathbb{F}_{q^m}}\chi'(s(f(x)+g(y)))\notag\\
        &=q^{m}+\frac{1}{q^m}\sum_{s\in\mathbb{F}_q^*}\sum_{(x,y)\in\mathbb{F}_{q^m}^2}\chi'(s(f(x)+g(y)))\notag\\
        &=q^{m}+\frac{1}{q^m}\sum_{s\in\mathbb{F}_q^*}\sum_{x\in\mathbb{F}_{q^m}}\chi'(sf(x))\sum_{y\in\mathbb{F}_{q^m}}\chi'(sg(y))\notag\\
        &=q^{m}+\frac{1}{q^m}\overline{\Phi}_{f,g,h}.\notag
    \end{align}
\end{proof}
Now we consider the weights $\text{wt}(\textbf{c}_{a,b,c})$, that is, the weights of the codewords in the subfield code $\mathcal{C}_{f,g}^{(q)}$. As before, let $\delta:\mathbb{F}_{q^m}\rightarrow \{0,1\}$ be a functions such that 
\begin{equation}\label{Vectorial delta}
    \delta(x)=\begin{cases}
        0 \,\,\,\,\,\text{if} \,\,\text{Tr}_{q^m/q}(x)=0\\
        1 \,\,\,\,\,\text{otherwise}.
    \end{cases}\notag
\end{equation}

\begin{lemma}\label{lemma: vectorial weights}
    For any $(a,b,c)\in\mathbb{F}_q\times\mathbb{F}_{q^m}^2$, the weight of the codeword $\textbf{c}_{a,b,c}$ in $\mathcal{C}_{f,g}^{(q)}$ is given by:
    $$\text{wt}(\textbf{c}_{a,b,c})=\begin{cases}
        0\,\,\,\,\,\,\,\,\,\,\,\,\,\,\,\,\,\,\,\,\,\,\,\,\,\,\,\,\,\,\,\,\,\,\,\,\,\,\,\,\,\,\,\,\,\,\,\,\,\,\,\,\,\,\,\,\,\,\,\,\,\,\,\,\,\,\,\,\,\,\,\,\,\,\,\,\,\,\,\,\,\,\,\,\,\,\,\,\,\,\,\,\,\,\,\,\,\,\,\,\,\,\,\,\,\,\,\,\,\,\,\,\,\,\,\,\,\,\,\,\,\,\,\,\,\,\,\,\text{if} \,\,a=b=c=0,\\
        q^{m}+\frac{1}{q^m}\overline{\Phi}_{f,g} \,\,\,\,\,\,\,\,\,\,\,\,\,\,\,\,\,\,\,\,\,\,\,\,\,\,\,\,\,\,\,\,\,\,\,\,\,\,\,\,\,\,\,\,\,\,\,\,\,\,\,\,\,\,\,\,\,\,\,\,\,\,\,\,\,\,\,\,\,\,\,\,\,\,\,\,\,\,\,\,\,\,\,\,\,\,\,\,\,\,\,\,\,\,\,\,\, \text{if}\,\, a\not =0\,\,\text{and}\,\,b=c=0\\
        \delta(b)+q^{m-1}(q-1)+\frac{1}{q^{m+1}}\big[ (q-1)\overline{\Phi_{f,g}}-\overline{\Lambda}_{a,b,c}^{f,g} \big]\,\,\,\,\,\,\text{otherwise}.  
    \end{cases}$$
\end{lemma}
\begin{proof}
    For $a\in\mathbb{F}_q^*$ and $b,c\in\mathbb{F}_{q^m}$ we define the term $N_{a,b,c}$ as the cardinality of the following set.
    $$N_{a,b,c}=\#\{(x,y)\in\mathcal{D}\,:\,a+\text{Tr}_{q^m/q}(bx+cy)=0\}.$$
    Clearly $N_{0,0,0}=\#\mathcal{D}$ and $N_{a,0,0}=0$ for $a\in\mathbb{F}_q^*$. For $b,c$ not all zero, we have:
    \begin{align}
        qN_{a,b,c}&=\sum_{(x,y)\in\mathcal{D}}\sum_{s\in\mathbb{F}_q}\chi(s(a+\text{Tr}_{q^m/q}(bx+cy)))\notag\\
    &=\#\mathcal{D}+\sum_{(x,y)\in\mathcal{D}}\sum_{s\in\mathbb{F}_q^*}\chi(s(a+\text{Tr}_{q^m/q}(bx+cy)))\notag\\
&=\#\mathcal{D}+\sum_{(x,y)\in\mathcal{D}}\sum_{s\in\mathbb{F}_q^*}\chi(sa)\chi'(sbx+scy)\notag\\
    &=\#\mathcal{D}+\sum_{s\in\mathbb{F}_q^*}\chi(sa)\sum_{(x,y)\in\mathbb{F}_{q^m}^2}\big(\frac{1}{q^m}\sum_{\omega\in\mathbb{F}_{q^m}}\chi'(\omega(f(x)+g(y))\big )\chi'(sbx+scy)\notag\\
&=\#\mathcal{D}+\frac{1}{q}\sum_{s\in\mathbb{F}_q^*}\chi(sa)\sum_{(x,y)\in\mathbb{F}_{q^m}^2}\chi'(sbx+scy)+\frac{1}{q^m}\overline{\Lambda}_{a,b,c}^{f,g}\notag\\
&=\#\mathcal{D}+\frac{1}{q^m}\overline{\Lambda}_{a,b,c,d}^{f,g,h}.\notag
    \end{align}
    The last equality is due to the fact that, for $s\in\mathbb{F}_q^*$ and $b,c$ not all zero:
    $$\sum_{(x,y)\in\mathbb{F}_{q^m}^2}\chi'(sbx+scy)=$$
    $$=\big(\sum_{x\in\mathbb{F}_{q^m}}\chi'(sbx)\big)\big(\sum_{y\in\mathbb{F}_{q^m}}\chi'(scy)\big)=0.$$
    Finally, from the equality 
    \begin{equation}\label{vectorial weight formula}
        \text{wt}(\textbf{c}_{a,b,c})=\delta(b)+\#\mathcal{D}-N_{a,b,c}
    \end{equation}
    we derive our statement.
\end{proof}

Since $\mathcal{C}^*_{f,g}$ is the punctured code of $\mathcal{C}_{f,g}$ at the first coordinate, from the two last lemmas we can deduce similar results for the punctured codes.

\begin{corollary}\label{corollary: vectorial parameters of the punctured codes}
    The length of the codes $\mathcal{C}_{f,g}^*$ and $\mathcal{C}_{f,g}^{*(q)}$ is $n=q^{m}+\frac{1}{q^m}\overline{\Phi}_{f,g}$. In addition,
    $$\text{wt}(\textbf{c}^*_{a,b,c})=\text{wt}(\textbf{c}_{a,b,c})-\delta(b),$$
    where the codeword $\textbf{c}^*_{a,b,c}$ is obtained from $\textbf{c}^*_{a,b,c}$ by deleting the first entry. 
\end{corollary}
\begin{remark}
    It is interesting to notice that it is possible to express the value of the parameter $\overline{\Lambda}_{a,b,c}^{f,g}$ in terms of the Walsh transform of the involved functions:
    \begin{align}
    \overline{\Lambda}_{a,b,c}^{f,g}&=\sum_{s\in\mathbb{F_q^*}}\chi(sa)\sum_{\omega\in\mathbb{F}_{q^m}^*}\sum_{(x,y)\in\mathbb{F}_{q^m}^2}\chi'(\omega f(x)+\omega g(y))\chi '(sbx+scy)\notag\\
    &=\sum_{s\in\mathbb{F_q^*}}\chi(sa)\sum_{\omega\in\mathbb{F}_{q^m}^*}\sum_{x\in\mathbb{F}_{q^m}}\chi'(\omega f(x)+sbx)\sum_{y\in\mathbb{F}_{q^m}}\chi '(\omega g(y)+scy)\notag\\
    &=\sum_{s\in\mathbb{F_q^*}}\chi(sa)\sum_{\omega\in\mathbb{F}_{q^m}^*} \mathcal{W}_f(\omega,sb)\mathcal{W}_g(\omega,sc)\notag\\
    &=\sum_{s\in\mathbb{F_q^*}}\chi(sa)\left (\sum_{\omega\in\mathbb{F}_{q^m}} \mathcal{W}_f(\omega,sb)\mathcal{W}_g(\omega,sc)- \mathcal{W}_f(0,sb)\mathcal{W}_g(0,sc)\right)\notag\\
    &=\sum_{s\in\mathbb{F_q^*}}\chi(sa) \sum_{\omega\in\mathbb{F}_{q^m}} \mathcal{W}_f(\omega,sb)\mathcal{W}_g(\omega,sc),
\end{align}
if at least one of $b$ and $c$ is non-zero.
\\
In particular, in the binary case, we have the equality:
\begin{align}
    \overline{\Lambda}_{a,b,c}^{f,g}&=(-1)^a \sum_{\omega\in\mathbb{F}_{2^m}} \mathcal{W}_f(\omega,b)\mathcal{W}_g(\omega,c)\notag\\&=(-1)^a \sum_{\omega\in\mathbb{F}_{2^m}} \mathcal{W}_{f^{-1}}(b, \omega)\mathcal{W}_g(\omega,c)\notag\\&=2^m(-1)^a\mathcal{W}_{f^{-1}\circ g}(b,c),\label{equation: Lambda in the vectorial framework}
\end{align}
if at least one of $b$ and $c$ is non-zero and $f$ is invertible.
\end{remark}
Since the Walsh transform is involved, we can analyze the codes within the general framework of plateaued functions.
\begin{thm}\label{thm: parameters in the plateaued case}
    Let $f$ and $g$ be two functions from $\mathbb{F}_{2^m}$ to $\mathbb{F}_{2^m}$, with $f$ invertible and such that the function $h:=f^{-1}\circ g$ is $s$-plateaued and normalized, that is, $h(0)=0$. Then $\mathcal{C}_{f,g}^{*(2)}$ is a four-weights $[2^{m}, 2m+1, 2^{m-1}-2^\frac{m+s-2}{2}]$ binary linear code, with the following weight distribution:
\begin{center}
\begin{tabular}{c|c}
\hline Weight & Multiplicity \\
\hline 0 & 1 \\
$2^{m-1}-2^{\frac{m+s-2}{2}}$ & $2^{2m-s}-2^{m-s}$ \\
$2^{m-1}$ & $(2^{m+1}-2)(2^m-2^{m-s}+1)$ \\
$2^{m-1}+2^{\frac{m+s-2}{2}}$ & $2^{2m-s}-2^{m-s}$ \\
$2^{m}$ & $1$ \\
\hline
\end{tabular}
\end{center}
For $m+s\ge 6$, the code is also doubly-even and hence self-orthogonal.\\
For $s=1$, the dual code $\mathcal{C}_{f,g}^{*(q)\perp}$ is a $[2^{m}, 2^{m}-2m-1,6]$ binary linear code that is almost dimensionally optimal with respect to the sphere packing bound. For $s>1$, the dual code $\mathcal{C}_{f,g}^{*(q)\perp}$ is a $[2^{m}, 2^{m}-2m-1,4]$ binary linear code.
\end{thm}

\begin{proof}
    From Corollary \ref{corollary: vectorial parameters of the punctured codes}, since $f$ is invertible we obtain that the length of the code is equal to $n=2^{m}$ and, combining the results of Corollary \ref{corollary: vectorial parameters of the punctured codes} and Lemma \ref{lemma: vectorial weights}, we can compute the weights of the codewords in the following way. For $(a,b,c)\in\mathbb{F}_2\times \mathbb{F}_{2^m}^2$:
    \begin{align*}
        \text{wt}(\textbf{c}_{a,b,c,d})&=\begin{cases}\notag
            0 \,\,\,\,\,\,\,\,\,\,\,\,\,\,\,\,\,\,\,\,\,\,\,\,\,\,\,\,\,\,\,\,\,\,\,\,\,\,\,\,\,\,\,\text{if} \,\,\,a=b=c=0\\
            2^{m} \,\,\,\,\,\,\,\,\,\,\,\,\,\,\,\,\,\,\,\,\,\,\,\,\,\,\,\,\,\,\,\,\,\,\,\,\,\,\,\text{if} \,\,\,a\not=0, b=c=0\\
            2^{m-1}-2^{\frac{m+s-2}{2}} \,\,\,\,\,\,\text{if} \,\,\, (-1)^aW_{h}(b,c)=2^\frac{m+s}{2}\\
            2^{m-1}+2^{\frac{m+s-2}{2}} \,\,\,\,\,\,\text{if}\,\,\, (-1)^aW_{h}(b,c)=-2^\frac{m+s}{2}\\
            2^{m-1} \,\,\,\,\,\,\,\,\,\,\,\,\,\,\,\,\,\,\,\,\,\,\,\,\,\,\,\,\,\,\,\,\text{if}\,\,\,W_{h}(b,c)=0        
            \end{cases}\\
        &=\begin{cases}
            0 \,\,\,\,\,\,\,\,\,\,\,\,\,\,\,\,\,\,\,\,\,\,\,\,\,\,\,\,\,\,\,\,\,\,\,\,\,\,\,\,\,\,\text{with 1 time}\\
            2^{m} \,\,\,\,\,\,\,\,\,\,\,\,\,\,\,\,\,\,\,\,\,\,\,\,\,\,\,\,\,\,\,\,\,\,\,\,\,\,\text{with 1 time}\\
            2^{m-1}-2^{\frac{m+s-2}{2}} \,\,\,\,\,\,\text{with} \,\,\,2^{2m-s}-2^{m-s}\,\,\,\text{times}\\
            2^{m-1}+2^{\frac{m+s-2}{2}} \,\,\,\,\,\,\text{with} \,\,\,2^{2m-s}-2^{m-s}\,\,\,\text{times}\\
            2^{m-1} \,\,\,\,\,\,\,\,\,\,\,\,\,\,\,\,\,\,\,\,\,\,\,\,\,\,\,\,\,\,\,\,\,\text{with}\,\,\,(2^{m+1}-2)(2^m-2^{m-s}+1)\,\,\,\text{times}\\
        \end{cases},
    \end{align*}
    which yields the code's weight distribution.\\\\
To determine the minimum distance of the dual code, we employ the Pless power moments identities. For $s=1$, we can use numerical computation to solve the system of equations generated by the first $r$ moments (for $r=1, \dots, 5$). By substituting the specific parameters $n$ and $k$ for $s=1$, the system yields the solution:
$$
A_1 = A_2 = A_3 = A_4 = A_5 = 0.
$$
The first non-zero solution appears at $w=6$, proving that $d=6$.
Analogously, for $s > 1$, the parameters $n$ and $k$ scale, altering the right-hand side of the moment identities. Evaluating the system for $r=1, \dots, 4$, we find that the constraints are satisfied with:
$$
A_1 = A_2 = A_3 = 0, \quad \text{but} \quad A_4 > 0.
$$
Specifically, the fourth-moment equation implies the existence of codewords of weight 4, thereby establishing that $d = 4$.
\end{proof}
\begin{example}
As a first case, we have the Almost Bent one, where $s=1$. Here are some known classes of Almost Bent power functions $f(x) = x^d$ over $\mathbb{F}_{2^m}$, where $m$ is odd.
\begin{table}[ht]
\centering
\renewcommand{\arraystretch}{1.8} 
\label{tab:ab_functions_refs}
\begin{tabular}{@{}llll@{}}
\toprule
\textbf{Name} & \textbf{Exponent} $d$ & \textbf{Conditions} & \textbf{References} \\ \midrule
Gold & $2^i + 1$ & $\gcd(i, m) = 1, \; 1 \leq i \leq \frac{m-1}{2}$ & \cite{Gold1968, Nyberg1994} \\
Kasami & $2^{2i} - 2^i + 1$ & $\gcd(i, m) = 1, \; 2 \leq i \leq \frac{m-1}{2}$ & \cite{Dobbertin1999, Kasami1971} \\
Welch & $2^{\frac{n-1}{2}} + 3$ & -- & \cite{Canteaut2000, ref6, Dobbertin1999Welch} \\
Niho & $2^{\frac{m-1}{2}} + 2^{\frac{m-1}{4}} - 1$ & $m \equiv 1 \pmod 4$ & \cite{Dobbertin1999Niho, Hollman2001} \\
Niho & $2^{\frac{m-1}{2}} + 2^{\frac{3m-1}{4}} - 1$ & $m \equiv 3 \pmod 4$ & \cite{Dobbertin1999Niho, Hollman2001} \\ \bottomrule
\end{tabular}
\end{table}
\\As a Corollary of Theorem \ref{thm: parameters in the plateaued case}, we can get the parameters of the code in the Almost Bent case. 
\begin{corollary}\label{corollary: parameters in the almost bent case}. 
    Let $f$ and $g$ be two functions from $\mathbb{F}_{2^m}$ to $\mathbb{F}_{2^m}$, with $f$ invertible and such that the function $h:=f^{-1}\circ g$ is almost bent and normalized, that is, $h(0)=0$. Then $\mathcal{C}_{f,g}^{*(2)}$ is a four-weights $[2^{m}, 2m+1, 2^{m-1}-2^\frac{m-1}{2}]$ binary linear code, with the following weight distribution:
\begin{center}
\begin{tabular}{c|c}
\hline Weight & Multiplicity \\
\hline 0 & 1 \\
$2^{m-1}-2^{\frac{m-1}{2}}$ & $2^{2m-1}-2^{m-1}$ \\
$2^{m-1}$ & $(2^{m+1}-2)(2^m-2^{m-1}+1)$ \\
$2^{m-1}+2^{\frac{m-1}{2}}$ & $2^{2m-1}-2^{m-1}$ \\
$2^{m}$ & $1$ \\
\hline
\end{tabular}
\end{center}
For $m\ge 5$ the code is also doubly-even and hence self-orthogonal.\\
The dual code $\mathcal{C}_{f,g}^{*(q)\perp}$ is a $[2^{m}, 2^{m}-2m-1,6]$ binary linear code that is almost dimensionally optimal with respect to the sphere packing bound.
\end{corollary}
\end{example}
Theorem \ref{thm: parameters in the plateaued case} introduces a family of few-weight codes whose duals are almost dimensionally optimal with respect to the Sphere Packing Bound in the Almost Bent case. These codes are parameterized by $s$-plateaued functions, a crucial class of functions in cryptographic applications, and it is worth noting that the weights of the codewords depend on the amplitude $s$ of the function. In the following Propositions, we recall two main constructions of vectorial plateaued functions (see also \cite{Mesnager-book,SM-al-2019}).

\begin{prop}[\emph{Primary construction}, \cite{Carlet-book}]
    Let $m$ be a positive integer, $\pi$ a permutation of $\mathbb{F}_{2^m}$, and $\phi, \psi$ two functions from $\mathbb{F}_{2^m}$ to $\mathbb{F}_{2^m}$. Let $i$ be an integer co-prime with $m$. Then, the function $F : \mathbb{F}_{2^m} \times \mathbb{F}_{2^m} \to \mathbb{F}_{2^m} \times \mathbb{F}_{2^m}$ defined by
\[
F(x, y) = \left( x\pi(y) + \phi(y), \, x(\pi(y))^{2^i} + \psi(y) \right)
\]
is plateaued, but does not have a single amplitude.
\end{prop}

\begin{prop}[\emph{Secondary construction}, \cite{Carlet-book}]
    Let $r, s, t, p$ be positive integers. Let $F$ be a plateaued $(r, t)$-function and $G$ be a plateaued $(s, p)$-function. Then, the function
\[
H(x, y) = (F(x), G(y)), \quad \text{where } x \in \F_2^r, \, y \in \F_2^s,
\]
is a plateaued $(r + s, t + p)$-function.

\noindent Indeed, for every $(a, b) \in \F_2^r \times \F_2^s$ and every $(u, v) \in \F_2^t \times \F_2^p$, we have:
\[
\mathcal{W}_H((a, b), (u, v)) = \mathcal{W}_F(a, u) \mathcal{W}_G(b, v).
\]
Note that this relation holds even if $u$ or $v$ is null. However, such a function $H$ is never of single amplitude, except in the trivial case when $F$ and $G$ are affine.
\end{prop}

\begin{remark}
    Given two functions $f$ and $g$, the goal of our construction is to ensure that the composition $h = f^{-1} \circ g$ is $s$-plateaued. In practice, it is more convenient to select an $s$-plateaued function $h$ first, fix a generic permutation polynomial $f$, and then define $g$ as $g = f \circ h$ to satisfy the relation.

For example, if we let $h$ be an $s$-plateaued function and $f$ be the identity function, then we must have $g = h$. If $g$ is invertible, the generator matrix has the following form.
\begin{equation}\label{equation: first generic construction case}
G_{\text{Id},g}^*
=
\begin{pmatrix}
1 \\
x \\
g^{-1}(x)
\end{pmatrix}_{x \in \mathbb{F}_{2^m}} .
\end{equation}
This code can be seen as a particular case of the \emph{first generic construction} framework.
The first generic construction is obtained by considering a code $C(f)$ over $\mathbb{F}_p$ involving a polynomial $f$ from $\mathbb{F}_q$ to $\mathbb{F}_q$, where $q = p^m$ (see \cite{ChapterMesnager}). Such a code is defined by
\[
C(f) = \left\{ \textbf{c} = \left( \mathrm{Tr}_{q/p}(af(x) + bx) \right)_{x \in \mathbb{F}_q} \;\middle|\; a \in \mathbb{F}_q, \, b \in \mathbb{F}_q \right\}.
\]
The resulting code $C(f)$ from $f$ is a linear code over $\mathbb{F}_p$ of length $q$ and its dimension is upper bounded by $2m$, which is reached when the nonlinearity of the vectorial function $f$ is larger than $0$.
In general, $\mathcal{C}(f)$ is the subfield code over $\mathbb{F}_p$ of the code $\mathcal{L}(f)$ generated by the following matrix 
\begin{equation}\label{equation: first generic construction}
G_f
=
\begin{pmatrix}
x \\
f(x)
\end{pmatrix}_{x \in \mathbb{F}_{q}} .
\end{equation}
\end{remark}
This proves the proposition linking our proposed construction to the first generic construction of linear codes.
\begin{prop}\label{prop: link with the first generic construction}
    Let $g$ be an invertible function from $\mathbb{F}_{2^m}$ to itself and $\text{Id}$ be the identity function of $\mathbb{F}_{2^m}$. 
    Then $\mathcal{C}^*_{\text{Id},g}$ is the augmented version of the code $\mathcal{L}(g^{-1})$, as the row of all ones is added to the generator matrix. Hence, $\mathcal{C}^{*(2)}_{\text{Id},g}$ is an augmented version of the first generic construction code $\mathcal{C}(g^{-1})$: 
    $$\mathcal{C}(g^{-1})=\{\textbf{c}^*_{a,b,c}\in\mathcal{C}^{*(2)}_{\text{Id},g}, a=0\}.$$
\end{prop}
In this way, we can obtain the parameters of the first generic construction code from the inverse of an s-plateaued function.
\begin{corollary}\label{corollary: first generic construction from plateuaed}
    Let $g$ be an invertible s-plateaued normalized function from $\mathbb{F}_{2^m}$ to itself. The first generic construction code $\mathcal{C}(g^{-1})$ is a three-weights $[2^m, 2m, 2^{m-1}-2^\frac{m+s-2}{2}]$ binary linear code with the following weight distribution:
    \begin{center}
\begin{tabular}{c|c}
\hline Weight & Multiplicity \\
\hline 0 & 1 \\
$2^{m-1}-2^{\frac{m+s-2}{2}}$ & $(2^{m}-1)(2^{m-s-1}+2\frac{m-s-2}{2})$ \\
$2^{m-1}$ & $2^{2m}+2^{m-s}-2^{2m-s}-1$ \\
$2^{m-1}+2^{\frac{m+s-2}{2}}$ & $(2^{m}-1)(2^{m-s-1}-2\frac{m-s-2}{2})$ \\
\hline
\end{tabular}
\end{center}
In addition, for $m\ge s+4$, the code is minimal. For $m+s\ge 6$, the code is also doubly-even, that is, the weight of every codeword is divisible by $4$, and hence self-orthogonal.
\end{corollary}
\begin{proof}
    Thanks to Proposition \ref{prop: link with the first generic construction}, we have that $\mathcal{C}(g^{-1})$ is formed by the codewords $\textbf{c}^*_{a,b,c}\in\mathcal{C}^{*(2)}_{a,b,c}$ with the condition $a=0$. Hence, we can follow the path of the proof of Theorem \ref{thm: parameters in the plateaued case}: combining the results of Corollary \ref{corollary: vectorial parameters of the punctured codes} and Lemma \ref{lemma: vectorial weights}, we can compute the weights of the codewords in the following way. For $(b,c)\in \mathbb{F}_{2^m}^2$:
    \begin{align*}
        \text{wt}(\textbf{c}_{a,b,c,d})&=\begin{cases}
            0 \,\,\,\,\,\,\,\,\,\,\,\,\,\,\,\,\,\,\,\,\,\,\,\,\,\,\,\,\,\,\,\,\,\,\,\,\,\,\,\,\,\,\,\,\,\,\text{if} \,\,\,b=c=0\\
            2^{m-1}-2^{\frac{m+s-2}{2}} \,\,\,\,\,\,\,\,\,\text{if} \,\,\, W_{h}(b,c)=2^\frac{m+s}{2}\\
            2^{m-1}+2^{\frac{m+s-2}{2}} \,\,\,\,\,\,\,\,\,\text{if}\,\,\, W_{h}(b,c)=-2^\frac{m+s}{2}\\
            2^{m-1} \,\,\,\,\,\,\,\,\,\,\,\,\,\,\,\,\,\,\,\,\,\,\,\,\,\,\,\,\,\,\,\,\,\,\,\text{if}\,\,\,W_{h}(b,c)=0      
            \end{cases}\\
        &=\begin{cases}
            0 \,\,\,\,\,\,\,\,\,\,\,\,\,\,\,\,\,\,\,\,\,\,\,\,\,\,\,\,\,\,\,\,\,\,\,\,\,\,\,\,\,\,\,\,\,\text{with 1 time}\\
            2^{m-1}-2^{\frac{m+s-2}{2}} \,\,\,\,\,\,\,\,\,\text{with} \,\,\,(2^{m}-1)(2^{m-s-1}+2\frac{m-s-2}{2})\,\,\,\text{times}\\
            2^{m-1}+2^{\frac{m+s-2}{2}} \,\,\,\,\,\,\,\,\,\text{with} \,\,\,(2^{m}-1)(2^{m-s-1}-2\frac{m-s-2}{2})\,\,\,\text{times}\\
            2^{m-1} \,\,\,\,\,\,\,\,\,\,\,\,\,\,\,\,\,\,\,\,\,\,\,\,\,\,\,\,\,\,\,\,\,\,\,\text{with}\,\,\,2^{2m}+2^{m-s}-2^{2m-s}-1\,\,\,\text{times}
        \end{cases},
    \end{align*}
    which yields the code's weight distribution.\\\\
    We examine the minimality of the code using the sufficient condition derived from the Ashikhmin-Barg bound \cite{Barg}. For a linear code over a field $\mathbb{F}_q$, the codewords are minimal if the ratio of the minimum weight $w_{\min}$ to the maximum weight $w_{\max}$ satisfies:
\begin{equation}
\frac{w_{\min}}{w_{\max}} > \frac{q-1}{q}.
\end{equation}
Since the proposed code is binary ($q=2$), this condition simplifies to:
$$
\frac{w_{\min}}{w_{\max}} > \frac{1}{2}.
$$
Substituting the explicit weights $w_{\min} = 2^{m-1} - 2^\frac{m+s-2}{2}$ and $w_{\max} = 2^{m-1} + 2^\frac{m+s-2}{2}$, we obtain:
$$
\frac{2^{m-1} - 2^\frac{m+s-2}{2}}{2^{m-1} + 2^\frac{m+s-2}{2}} > \frac{1}{2}.
$$
After some computation, we have:
$$
\frac{m-s}{2} > \log_2(3) \approx 1.585.
$$
Therefore, the code is minimal for all integers $m$ such that $m\ge s+4$.
\end{proof}

\section{Applications in the context of quantum CSS$_T$ codes}


Quantum error-correcting codes are a fundamental ingredient in fault-tolerant
quantum computation \cite{NC}.
Among the most prominent constructions are Calderbank--Shor--Steane (CSS) codes,
which are built from pairs of classical linear codes
\cite{CP,St}.

Let $C_X$ and $C_Z$ be two linear codes over $\mathbb{F}_p$ of length $N$ such that
\[
C_Z^\perp \subseteq C_X .
\]
Then the pair $(C_X,C_Z)$ defines a CSS quantum code encoding
\[
k = \dim(C_X) + \dim(C_Z) - N
\]
logical qudits into $N$ physical qudits.
The minimum distance of the code is given by
\[
d = \min\big\{ \mathrm{wt}(c) : c\in C_X\setminus C_Z^\perp
\ \text{or}\ c\in C_Z\setminus C_X^\perp \big\}.
\]

One of the main advantages of CSS codes lies in the separation of $X$-type and
$Z$-type errors, which allows for efficient decoding procedures and a
transparent algebraic structure \cite{NC}.

\subsection{Transversal gates and fault tolerance}

A central problem in fault-tolerant quantum computation is the realization of a
universal set of logical gates implemented in a fault-tolerant manner.
A particularly desirable class of gates are \emph{transversal gates}, which act
independently on each physical qudit and therefore do not propagate errors.

However, the Eastin--Knill theorem states that no quantum error-correcting code
can admit a universal set of transversal gates \cite{EK}.
As a consequence, one seeks codes admitting transversal implementations of
specific non-Clifford gates, which can then be supplemented by additional
techniques such as magic-state distillation \cite{BK,Campbell}.

\subsubsection{CSS$_T$ codes in the binary case}

In the binary setting, that is $p=2$, the most important non-Clifford gate is the
$T$-gate, defined by
\[
T = \mathrm{diag}(1,e^{i\pi/4}).
\]
A CSS code is said to be a \emph{CSS$_T$ code} if the logical $T$-gate can be
implemented transversally.

Throughout this work, the term CSS$_T$ code refers specifically to the binary
case. For odd prime dimensions, we use the more general notion of CSS phase codes
admitting transversal phase gates.

For binary CSS$_T$ codes, a sufficient condition for
transversal $T$-gates is the \emph{divisibility condition}
\begin{equation}\label{eq:divisible}
\mathrm{wt}(c) \equiv 0 \pmod{4}
\qquad \text{for all } c\in C_X .
\end{equation}
Codes satisfying~\eqref{eq:divisible} are often referred to as
\emph{doubly-even} codes \cite{BH}.

This condition ensures that the phase accumulated by applying $T$ transversally
on all physical qubits cancels appropriately in the codespace, resulting in a
well-defined logical operation \cite{BH}.

\subsubsection{CSS phase codes in odd characteristic}

For qudits of odd prime dimension $p>2$, the $T$-gate is replaced by higher-order
phase gates.
For an integer $k$ dividing $p-1$, define the phase gate
\[
R_k = \mathrm{diag}\big( \omega^{x^k} \big)_{x\in\mathbb{F}_p},
\qquad \omega = e^{2i\pi/p}.
\]
A CSS code over $\mathbb{F}_p$ is said to admit a \emph{transversal phase gate
$R_k$} if the logical action of $R_k$ can be implemented by applying $R_k$
independently to each physical qudit \cite{Campbell,GH}.

In this context, the divisibility condition~\eqref{eq:divisible} generalizes to
the \emph{moment divisibility condition}
\begin{equation}\label{eq:phase}
\sum_{i=1}^N c_i^k \equiv 0 \pmod{p}
\qquad \text{for all } c=(c_1,\dots,c_N)\in C_X .
\end{equation}
Unlike the binary case, this condition cannot be reduced to a simple
Hamming-weight constraint and genuinely reflects higher-order global
correlations in the codewords \cite{Campbell}.

\subsection{CSS phase codes from plateaued functions}

Most known constructions of CSS$_T$ and CSS phase codes are derived from Reed--Muller codes
or their generalizations.
While powerful, these constructions are highly structured and limited in diversity.

In this work, we propose an alternative approach based on linear codes constructed from
vectorial plateaued functions.
The flatness of the Walsh spectrum of such functions leads to highly regular weight and moment distributions,
making them natural candidates for CSS codes admitting transversal phase gates. This perspective establishes a new bridge between Boolean and $p$-ary function theory,
classical coding theory, and fault-tolerant quantum computation.
\\\\
Unless otherwise specified, $\mathrm{Tr}$ denotes the absolute trace
from $\mathbb{F}_{p^m}$ to $\mathbb{F}_p$. Let $p$ be a prime and let
\[
F:\mathbb{F}_{p^n} \longrightarrow \mathbb{F}_{p^m}
\]
be a vectorial $s$-plateaued function.
Denote by $C_{_F}$ the $p$-ary linear code associated with $F$ via the standard \emph{second generic construction} (see \cite{ChapterMesnager}), and by $C_F^\perp$ its dual. Hence,
$$C_F=\{c_a=(\mathrm{Tr}(aF(x)))_{x\in\mathbb{F}_{p^n}}\,|\,a\in\mathbb{F}_{p^m}\}.$$
We consider the CSS construction defined by
\[
C_X = C_F, \qquad C_Z = C_F^\perp .
\]
Since $C_Z^\perp = C_F \subseteq C_X$, the CSS commutation condition is automatically satisfied.
\\\\
Let $k$ be a positive integer dividing $p-1$.
The CSS code admits a transversal phase gate $R_k$ if and only if, for every codeword
$c=(c_1,\dots,c_N)\in C_X$,
\[
\sum_{i=1}^N c_i^k \equiv 0 \pmod{p}.
\]
In the present setting, codewords of $C_F$ are of the form
\[
c_a = \big( \mathrm{Tr}(aF(x)) \big)_{x\in\mathbb{F}_{p^n}}, \qquad a\in\mathbb{F}_{p^m}.
\]
Hence, the transversal phase condition becomes

\begin{equation}\label{eq:phase-condition}
\sum_{x\in\mathbb{F}_{p^n}} \big( \mathrm{Tr}(aF(x)) \big)^k \equiv 0 \pmod{p},
\qquad \text{for all}\,\, a\in\mathbb{F}_{p^m}^*
\end{equation}
This condition depends only on the global distribution of the values of $F$ and is naturally amenable to
analysis via the spectral properties of plateaued functions.

\subsubsection{Binary case: CSS$_T$ codes}

For $p=2$, the relevant condition corresponds to the quadratic moment, that is $k=2$,
which governs the transversal implementation of the $T$-gate.
Since $x^2=x$ for all $x\in\mathbb{F}_2$, condition~\eqref{eq:phase-condition} reduces to
\[
\sum_{x\in\mathbb{F}_{2^n}} \mathrm{Tr}(aF(x)) \equiv 0 \pmod{2},
\]
which is equivalent to requiring that all codewords of $C_F$ have even weight. On the other hand, for binary CSS$_T$ codes, the doubly even condition must be satisfied.

Let $c_a$ be the codeword in $C_F$ parameterized by $a \in \mathbb{F}_{2^m}^*$. The Hamming weight of $c_a$ is the number of inputs $x \in \mathbb{F}_{2^n}$ for which the trace evaluates to $1$:
\begin{equation}
    \text{wt}(c_a) = \left| \{ x \in \mathbb{F}_{2^n} : \text{Tr}(aF(x)) = 1 \} \right|.
\end{equation}
Consider the sum of the additive characters over all inputs $x \in \mathbb{F}_{2^n}$:
\begin{equation}
    S = \sum_{x \in \mathbb{F}_{2^n}} (-1)^{\text{Tr}(aF(x))}.
\end{equation}
We can split this sum based on the value of the trace function. Let $N_0$ be the number of $x$ such that $\text{Tr}(aF(x)) = 0$, and $N_1 = \text{wt}(c_a)$ be the number of $x$ such that $\text{Tr}(aF(x)) = 1$. Since $N_0 + N_1 = 2^n$, we have $N_0 = 2^n - \text{wt}(c_a)$. Thus:
\begin{align*}
    S &= N_0 - N_1 \\
      &= (2^n - \text{wt}(c_a)) - \text{wt}(c_a) \\
      &= 2^n - 2 \cdot \text{wt}(c_a).
\end{align*}
Recall the definition of the Walsh transform $\mathcal{W}_F(a, b)$ with the input $b$ set to zero:
\begin{equation*}
    \mathcal{W}_F(a, 0) = \sum_{x \in \mathbb{F}_{2^n}} (-1)^{\text{Tr}(aF(x) + 0 \cdot x)} = \sum_{x \in \mathbb{F}_{2^n}} (-1)^{\text{Tr}(aF(x))} = S.
\end{equation*}
Substituting $\mathcal{W}_F(a, 0)$ into the expression for $S$:
\begin{equation*}
    \mathcal{W}_F(a, 0) = 2^n - 2 \cdot \text{wt}(c_a).
\end{equation*}
Solving for $\text{wt}(c_a)$, we obtain the final formula:
\begin{equation}
    \text{wt}(c_a) = 2^{n-1} - \frac{1}{2} \mathcal{W}_F(a, 0).
\end{equation}

\begin{corollary}[Binary CSS$_T$ codes]
Let $F:\mathbb{F}_{2^n}\to\mathbb{F}_{2^m}$ be an $s$-plateaued vectorial function such that
\[
n+s \ge 6.
\]
Then all codewords of $C_F$ have weight divisible by $4$, and the CSS code defined by
$(C_F,C_F^\perp)$ is a CSS$_T$ quantum code.
\end{corollary}
\begin{proof}
Let $a\in\mathbb{F}_{2^m}^*$.  
Since $F$ is an $s$-plateaued vectorial function, by definition $n+s$ is even. In addition, the Walsh spectrum satisfies
\[
\mathcal{W}_F(a,0)\in\{0,\pm 2^{(n+s)/2}\}.
\]
The Hamming weight of the corresponding codeword $c_a\in C_F$ is given by
\[
\mathrm{wt}(c_a)
=
2^{n-1}-\frac{1}{2}\mathcal{W}_F(a,0).
\]

If $\mathcal{W}_F(a,0)=0$, then $\mathrm{wt}(c_a)=2^{n-1}$, which is divisible by $4$ whenever $n\ge 3$. This condition is implied by $n+s\ge 6$, since by definition $0\le s\le n$.

If $\mathcal{W}_F(a,0)=\pm 2^{(n+s)/2}$, then
\[
\mathrm{wt}(c_a)
=
2^{n-1}\mp 2^{(n+s)/2-1}.
\]
The assumption $n+s\ge 6$ implies that $(n+s)/2-1$ is an integer
greater than or equal to $2$. Moreover, since $n-1\ge (n+s)/2-1$, both summands are
divisible by $4$, and therefore $\mathrm{wt}(c_a)\equiv 0\pmod{4}$.

Thus, all codewords of $C_F$ have Hamming weight divisible by $4$, and the
divisibility condition~\eqref{eq:divisible} is satisfied. Consequently, the CSS
code defined by $(C_F,C_F^\perp)$ admits a transversal $T$-gate and is a
CSS$_T$ quantum code.
\end{proof}

\begin{remark}
    Notice that almost all families of binary codes constructed in this paper are doubly even for sufficiently large values of $m$ and $s$, so that every term in the weight expression is divisible by $4$. Hence, they are self-orthogonal. In addition, using their dual codes, whose parameters have also been determined, it is possible to construct a CSS$_T$ quantum code. In particular, let $\mathcal{C}$ be such a code. Then the CSS code defined by $(\mathcal{C},\mathcal{C}^\perp)$ admits a transversal $T$-gate and is therefore a CSS$_T$ quantum code.
\end{remark}

\subsubsection{Ternary case: CSS phase codes from quadratic and Niho-type plateaued functions}

For $p=3$, we have $p-1=2$, and the only nontrivial choice is $k=2$.
The corresponding transversal gate is the ternary quadratic phase gate $R_2$.

Condition~\eqref{eq:phase-condition} becomes
\begin{equation}\label{eq:ternary-condition}
\sum_{x\in\mathbb{F}_{3^n}} \big( \mathrm{Tr}(aF(x)) \big)^2 \equiv 0 \pmod{3},
\qquad \text{for all}\,\, a\in \mathbb{F}_{3^m}^*.
\end{equation}
Unlike the binary case, this is not a weight divisibility condition but a genuine quadratic moment constraint.

\begin{corollary}[Ternary CSS phase codes]
Let $F:\mathbb{F}_{3^n}\to\mathbb{F}_{3^m}$ be a vectorial plateaued function satisfying
\[
\sum_{x\in\mathbb{F}_{3^n}} \big( \mathrm{Tr}(aF(x)) \big)^2 \equiv 0 \pmod{3}
\quad \text{for all } a\in\mathbb{F}_{3^m}^{*}
\]
Then the CSS code defined by $(C_F,C_F^\perp)$ admits a transversal ternary phase gate $R_2$.
\end{corollary}

\begin{proof}
By construction, codewords of $C_F$ are of the form
\[
c_a=\big(\mathrm{Tr}(aF(x))\big)_{x\in\mathbb{F}_{3^n}}, \qquad a\in\mathbb{F}_{3^m}.
\]
The hypothesis of the corollary states that, for all $a\neq 0$,
\[
\sum_{x\in\mathbb{F}_{3^n}} \big(\mathrm{Tr}(aF(x))\big)^2 \equiv 0 \pmod{3}.
\]
This is precisely the quadratic moment divisibility condition
\eqref{eq:ternary-condition}, which is necessary and sufficient for the existence
of a transversal quadratic phase gate $R_2$ acting on the CSS code with
$C_X=C_F$. Therefore, the CSS code defined by $(C_F,C_F^\perp)$ admits a transversal
ternary phase gate $R_2$.
\end{proof}

A vectorial quadratic function from $\mathbb{F}_{3^n}$ to $\mathbb{F}_{3^n}$ is plateaued whenever the associated symmetric bilinear forms
have constant rank. In particular, every quadratic form is plateaued, as pointed out by Carlet and Prouff in \cite{QuadraticCarlet}.

\begin{thm}
Let $n \ge 3$ be an integer, $F:\mathbb{F}_{3^n}\rightarrow\mathbb{F}_{3^n}$ be a quadratic function, including functions of Niho type, and $C_F$ be the ternary linear code defined by the trace map:
\[
C_F = \left\{ \big(\mathrm{Tr}_{3^n/3}(aF(x))\big)_{x \in \mathbb{F}_{3^n}} : a \in \mathbb{F}_{3^n} \right\}.
\]
Then, every codeword in $C_F$ satisfies the quadratic moment condition:
\[
\sum_{x\in\mathbb{F}_{3^n}} \big(\mathrm{Tr}_{3^n/3}(aF(x))\big)^2 \equiv 0 \pmod{3} \quad \text{for all } a \in \mathbb{F}_{3^n}.
\]
Consequently, the associated CSS code admits a transversal ternary phase gate $R_2$.
\end{thm}

\begin{proof}
Fix an arbitrary element $a \in \mathbb{F}_{3^n}$. We define the component function $f_a: \mathbb{F}_{3^n} \to \mathbb{F}_3$ by
\[
f_a(x) = \mathrm{Tr}_{3^n/3}(aF(x)).
\]
Since $F(x)$ is a quadratic function on $\mathbb{F}_{3^n}$ and the trace is a linear operator, the composition $f_a(x)$ is a quadratic form from $\mathbb{F}_{3^n}$ to $\mathbb{F}_3$.

We evaluate the sum of squares of the codeword components to verify the self-orthogonality. In the field $\mathbb{F}_3$, the elements satisfy $y^2 = 1$ if $y \neq 0$ and $y^2 = 0$ if $y = 0$. Thus, the sum of squares is congruent modulo 3 to the Hamming weight of the function $f_a$:
\[
\sum_{x\in\mathbb{F}_{3^n}} (f_a(x))^2 = \sum_{x : f_a(x) \neq 0} 1 = w_H(f_a).
\]
According to the standard classification of quadratic forms over finite fields of odd characteristic, the Hamming weight $w_H(f_a)$ is given by
\[
w_H(f_a) = 3^n - N_0,
\]
where $N_0$ is the number of zeros of $f_a$. The value of $N_0$ is of the form $3^{n-1} + \epsilon 3^{n-1-h}$, where $2h$ (or $2h+1$) is the rank of the form and $\epsilon \in \{0, \pm 1\}$ (see \cite{Mesnager-book}). Consequently, the weight $w_H(f_a)$ is a multiple of $3^{n-1-h}$.
Since the rank cannot exceed $n$, the smallest power of 3 dividing the weight is determined by the case of maximal rank.

\begin{itemize}
    \item If $n$ is odd, the maximal rank is $n-1$ (since symplectic forms over $\mathbb{F}_3$ have even rank), so the weight is a multiple of $3^{(n-1)/2}$. Since $n \ge 3$, $(n-1)/2 \ge 1$, ensuring divisibility by 3.
    \item If $n$ is even, the maximal rank is $n$ (bent functions), so the weight is a multiple of $3^{n/2 - 1}$. Since $n \ge 3$ implies $n \ge 4$ (for even $n$), we have $n/2 - 1 \ge 1$, ensuring divisibility by 3.
\end{itemize}

Therefore, for all $n \ge 3$, we have $w_H(f_a) \equiv 0 \pmod{3}$. This implies
\[
\sum_{x\in\mathbb{F}_{3^n}} (f_a(x))^2 \equiv 0 \pmod{3},
\]
which proves that the code is self-orthogonal with respect to the Euclidean inner product. This is the necessary and sufficient condition for the associated CSS code to admit a transversal implementation of the quadratic phase gate $R_2$.
\end{proof}

The above results demonstrate that vectorial plateaued functions provide a unified and natural framework
for constructing CSS quantum codes with transversal phase gates in both binary and ternary settings.
While the binary case recovers CSS$_T$ codes, the ternary construction yields genuinely new families
of CSS phase codes not derived from classical Reed--Muller structures.

\section*{Conclusion and remarks}
In this paper, we have presented two constructions of linear codes defined by specific subsets $D_1 \subset \mathbb{F}_{q^m}^3$ and $D_2 \subset \mathbb{F}_{q^m}^2$, involving scalar and vectorial functions over finite fields. By leveraging the properties of bent and, more generally, $s$-plateaued functions, we determined the parameters and weight distributions of their subfield codes.

A central result of our work is the detailed analysis of the binary case, where we established that the constructed codes exhibit few weights, specifically three, four, or five, depending on the nature of the underlying functions. We proved that when the defining function $h = f^{-1} \circ g$ is almost bent ($s = 1$), the dual codes are almost dimensionally optimal with respect to the sphere packing bound. Additionally, we extended our framework to a construction involving three $p$-ary functions $f, g, h$. For specific families of functions, we determined the complete weight distribution of the resulting codes and proved that, in the $s$-plateaued case with odd $s$, we have distance and dimension optimality with respect to the sphere packing bound. Also, in the bent case, the dual of the punctured code is almost dimensionally optimal with respect to the sphere packing bound. 

Furthermore, we demonstrated that our setting provides a new perspective on the first generic construction of linear codes. Specifically, we showed that for an invertible $s$-plateaued function $g$, the subfield code $\mathcal{C}^{*(2)}_{Id,g}$ serves as an augmented version of the first generic construction code $\mathcal{C}({g^{-1}})$. This interpretation allowed us to derive the weight distribution and minimality conditions for these codes. 

Finally, we explored applications of the proposed linear codes in fault-tolerant quantum computation. We showed that linear codes derived from vectorial plateaued functions
integrate naturally into the Calderbank--Shor--Steane framework and yield quantum codes
admitting transversal phase gates. In the binary case, suitably chosen $s$-plateaued functions lead to doubly-even codes,
thereby recovering CSS$_T$ codes with transversal $T$-gates from a unified Walsh spectral viewpoint.
Beyond characteristic two, and in particular in the ternary setting, the same construction
produces CSS phase codes admitting transversal quadratic phase gates.
These codes satisfy higher-order moment conditions inherent to large classes of quadratic
and Niho-type plateaued functions and do not arise from classical Reed--Muller constructions.

Altogether, this work highlights a deep connection between (vectorial) plateaued function theory, few-weight linear codes and fault-tolerant quantum computation as a versatile tool for designing both classical and quantum codes with controlled algebraic and fault-tolerance properties.

Our achievements open several avenues for future research:

\begin{itemize}
    \item \textbf{Extensions to $q$-ary functions:} Although our weight analysis focused on the binary case, the construction $\mathcal{C}_{f,g}$ naturally applies to $\mathbb{F}_q$. Future work could employ this framework for other families of vectorial functions over $\mathbb{F}_{q^m}$, particularly investigating the properties of almost bent functions to determine whether similar optimality results hold.

    \item \textbf{Generalized defining sets:} The defining set $D$ in our construction relies on the equation $f(x) + g(y) = 0$. Modifying this condition---for instance, by introducing additional linear or non-linear terms---could yield new families of linear codes with desirable parameters.

    \item \textbf{Multi-function constructions:} We restricted our attention primarily to two functions, $f$ and $g$, with a specific case for three functions in the non-vectorial framework. Extending the construction to include more generic functions could yield codes with higher dimensions or different weight spectra, which may be useful for specific error-correction applications.

    \item \textbf{Augmented first generic construction:} Our observation that $\mathcal{C}^{*(2)}_{Id,g}$ acts as an augmented first generic construction suggests a new tool for analyzing classical codes. Future studies could exploit this interpretation to revisit and characterize other well-known families of codes derived from cryptographic functions in the first generic construction framework.
\end{itemize}

\section*{Acknowledgements}
The authors are grateful to Anne Canteaut for insightful discussions and valuable suggestions following the first author's presentation at the C2 Conference, held in 2025 in Pornichet, France.

\end{document}